\documentclass[12pt]{iopart}

\usepackage{newtxtext}
\usepackage{color}
\usepackage{comment}
\usepackage{cite}
\usepackage{amssymb}
\usepackage{amsthm}
\usepackage{iopams}

\theoremstyle{definition}
\newtheorem{Def}{Definition}[section]

\newtheorem{thm}[Def]{Theorem}

\newtheorem{cor}[Def]{Corollary}
\newtheorem{eg}[Def]{Example}
\newtheorem{rem}[Def]{Remark}

\newcommand{\al}{\alpha}
\newcommand{\be}{\beta}

\newcommand{\del}{\delta}
\newcommand{\eps}{\epsilon}

\newcommand{\bfrac}[2]{\left(\frac{#1}{#2}\right)}
\newcommand{\od}[3][1]{
\ifnum #1=1\frac{\mathrm{d} #2}{\mathrm{d} #3}
\else\frac{\mathrm{d}^{#1} #2}{\mathrm{d} #3^{#1}}
\fi}
\newcommand{\pd}[3][1]{
\ifnum #1=1\frac{\partial #2}{\partial #3}
\else\frac{\partial^{#1} #2}{\partial #3^{#1}}
\fi}

\newcommand{\udis}{ultra-discrete}
\newcommand{\udisn}{ultra-discretization}
\newcommand{\ca}{cellular automaton}
\newcommand{\caa}{cellular automata}
\newcommand{\bbs}{box and ball system}
\newcommand{\lv}{Lotka-Volterra}

\usepackage{graphicx}

\newcommand{\tail}{\mathrm{tail}}
\newcommand{\head}{\mathrm{head}}
\newcommand{\floor}[1]{\left\lfloor{#1}\right\rfloor}

\begin{document}

\title
[A delay analogue of the BBS arising from the delay discrete LV equation]
{A delay analogue of the \bbs\ arising from the \udisn\ of the delay discrete \lv\ equation}
\author{Kenta Nakata$^1$, Kanta Negishi$^1$, Hiroshi Matsuoka$^1$ and Ken-ichi Maruno$^2$}
\address{$^1$~Department of Pure and Applied Mathematics, School of Fundamental Science and Engineering, Waseda University, 3-4-1 Okubo, Shinjuku-ku, Tokyo 169-8555, Japan}
\address{$^2$~Department of Applied Mathematics, Faculty of Science and Engineering, Waseda University, 3-4-1 Okubo, Shinjuku-ku, Tokyo 169-8555, Japan}
\ead{kennakaxx@akane.waseda.jp n.kanta@akane.waseda.jp matsuoka8@akane.waseda.jp kmaruno@waseda.jp}
\vspace{10pt}
\begin{indented}
\item[]{\today}
\end{indented}
\begin{abstract}
    A delay analogue of the \bbs\ (BBS) is presented.
    This new soliton \ca\ is constructed by the \udisn\ of the delay discrete \lv\ equation, which is an integrable delay analogue of the discrete \lv\ equation.
    %This delay BBS requires multiple time initial states for time evolution, thus it has various types of soliton patterns.
    Soliton patterns generated by this delay BBS are classified into normal solitons and abnormal solitons.
    Normal solitons have a clear relationship to the solitons of the BBS with $K$ kinds of balls.
    %can be discussed analytically, in particular, they
    On the other hand, abnormal solitons show various types of novel soliton patterns, which have not been observed in almost all known BBSs.
    We obtain them by numerical experiments, and then construct $\tau$-functions of them analytically in $1$-soliton cases.
\end{abstract}
\noindent{\it Keywords\/}:\ \ \bbs s, \udisn, delay discrete soliton equations

%\submitto{\jpa}

\begin{section}{Introduction}
\label{sec_intro}

The \bbs\ (BBS) proposed by Takahashi and Satsuma~\cite{Takahashi1,Takahashi2} is one of the most fundamental soliton \caa.
From the viewpoint of integrable systems, the BBS has a remarkable relationship to soliton equations.
Tokihiro \textit{et al}~\cite{Tokihiro3} applied a kind of limiting procedure called \udisn\ to the discrete \lv\ (LV) equation.
Then they showed this procedure yields the \udis\ LV equation, and it is directly connected to the time evolution equation of the BBS.
%Because of its importance in mathematics and mathematical physics,
After this discovery, various extended versions have been constructed and studied~\cite{Takahashi3,Tokihiro1,Hikami,Tokihiro2,Hatayama,Yura}.

Now we describe the detail of the \udisn\ of the LV equation as follows.
First, applying the following transformation and limit
\begin{equation*}
    v_{n}^{t}=\exp\bfrac{V_{n}^{t}}{\eps}\,,\qquad
    \del=\exp\bfrac{-1}{\eps}\,,\qquad
    \eps\to0\,,\qquad
    \left\{
    \begin{array}{cc}
        n+t&\to t\\
        t&\to n
    \end{array}
    \right.
\end{equation*}
to the discrete LV equation~\cite{Hirota1}
\begin{equation}
    \label{dislv}
    \frac{v_{n}^{t+1}}{v_{n}^{t}}
    =\frac{1+\del v_{n+1}^{t}}{1+\del v_{n-1}^{t+1}}\,,
    % (1+\del)f_{n}^{t+1}f_{n-1}^{t}
    % -\del f_{n+1}^{t}f_{n-2}^{t+1}
    % -f_{n}^{t}f_{n-1}^{t+1}
    % =0\,,\quad
    % v_{n}^{t}
    % =\frac{f_{n+1}^{t}f_{n-2}^{t+1}}{f_{n}^{t}f_{n-1}^{t+1}}\,,
\end{equation}
we obtain the \udis\ LV equation
\begin{equation}
    V_{n+1}^{t+1}-V_{n}^{t}
    =\max(V_{n}^{t+1}-1,0)-\max(V_{n+1}^{t}-1,0)\,.
    % F_{n+1}^{t+1}+F_{n}^{t-1}
    % =\max(F_{n}^{t+1}+F_{n+1}^{t-1}-1,
    % F_{n}^{t}+F_{n+1}^{t})\,.
\end{equation}
Then introducing the dependent variable $U_{n}^{t}$ by
\begin{equation*}
    V_{n}^{t}
    =\sum_{j=-\infty}^{n+1}U_{j}^{t}-\sum_{j=-\infty}^{n}U_{j}^{t+1}\,,
    % F_{n}^{t}=\sum_{i=t+1}^{\infty}\sum_{j=-\infty}^{n}U_{j}^{i}\,,
\end{equation*}
we obtain the time evolution equation of the BBS
\begin{equation}
    \label{bbs}
    U_{n}^{t+1}
    =\min\left(1-U_{n}^{t},
     \sum_{j=-\infty}^{n-1}U_{j}^{t}-\sum_{j=-\infty}^{n-1}U_{j}^{t+1}\right)\,.
\end{equation}
Note that we can also obtain the $N$-soliton solution of the BBS (\ref{bbs}) by applying the \udisn\ to the $N$-soliton solution of the discrete LV equation (\ref{dislv}).

The integrability of delay differential equations has been actively studied since the 1990s~\cite{Quispel,Levi,Gram,Ramani,Joshi1,Joshi2,Carstea,Viallet,Halburd,Berntson,Stokes,Sekiguchi}.
%In recent years, research on integrability of delay differential equations has been developing.
%Quispel \textit{et al}~\cite{Quispel}, Levi and Winternitz~\cite{Levi}, and Grammaticos \textit{et al}~\cite{Gram,Ramani} proposed delay-differential analogues of the Painlev\'e equations in different methods.
%Then several studies revealed their properties~\cite{Joshi1,Joshi2,Carstea,Viallet,Halburd,Berntson,Stokes}.
In our recent study~\cite{Tsunematsu,Nakata1}, some integrable delay analogues of soliton equations were proposed, such as a delay discrete LV equation
%In our recent study~\cite{Nakata1}, an integrable delay analogue of the discrete LV equation
\begin{equation}
    \label{dlydislv_nl}
    \frac{v_{n}^{t+1+\mu}v_{n-1}^{t}}{v_{n}^{t+\mu}v_{n-1}^{t+1}}
    =\frac{(1+\del v_{n+1}^{t+\mu})(1+\del v_{n-2}^{t+1})}{(1+\del v_{n}^{t+\mu})(1+\del v_{n-1}^{t+1})}\,,
\end{equation}
where the constant $\mu$ is considered as the delay parameter.
The delay discrete LV equation (\ref{dlydislv_nl}) has an $N$-soliton solution similar to the discrete LV equation (\ref{dislv}), and reduces to the delay LV equation
\begin{equation*}
    \frac{\mathrm{d}}{\mathrm{d}t}\log\frac{u_{n}(t+\tau)}{u_{n-1}(t-\tau)}
    =u_{n+1}(t+\tau)-u_{n}(t+\tau)-u_{n-1}(t-\tau)+u_{n-2}(t-\tau)
\end{equation*}
in the continuum limit $\del\to0\,,\ m\del=t\,,\ \mu\del=2\tau$.
%This equation (\ref{dlydislv_nl}) is named the delay discrete LV equation, and has the $N$-soliton solution similar to the discrete LV equation (\ref{dislv}).
%However, the effects of its ``delay'' have not yet been clarified.
%However, any effects of its ``delay'' are still not clarified.
%In other words, we have not yet observed any effects of the delay $\mu$.
%Drawing soliton solutions of equations (\ref{dislv}) and (\ref{dlydislv_nl}) at the same time, we can only find differences in the heights and velocities of their solitons.
%Therefore in this paper, we move from the discrete systems to the \udis\ systems to understand the effects of the delay.

The main purpose of this paper is to apply the \udisn\ to the delay discrete LV equation (\ref{dlydislv_nl}),
%and to investigate the obtained \udis\ systems.
to construct a delay analogue of the BBS, and then to visualize and analyze an effect of the delay.

This paper is organized as follows.
In section \ref{sec_cons}, we derive the delay \udis\ LV equation, and transform it into a delay analogue of the BBS.
%The $N$-soliton solution of the delay BBS is also obtained by the \udisn\ of that of the delay discrete LV equation
In section \ref{sec_obs}, several types of soliton interactions of this delay BBS are presented.
For the time evolution of the delay BBS, we need to set multiple time initial states, and this feature realizes to generate various soliton patterns.
For example, we can observe solitons whose speeds change periodically.
We define such solitons as normal solitons, and others as abnormal solitons.
Soliton patterns shown in abnormal solitons have not appeared in almost any type of BBS.
We construct exact $\tau$-functions of them in $1$-soliton cases.
In section \ref{sec_rule}, the rule for the time evolution of the delay BBS and soliton interactions are discussed.
Because the delay BBS contains various types of solitons, it is necessary to consider the rules for each type of soliton individually.
In this paper, we describe the rules for the case of normal solitons by only elementary operations of boxes and balls.
%In the case of normal solitons, only an elementary process of boxes and balls is required for time evolution, and furthermore, non-trivial conservation laws can be proven.
%In the case of abnormal solitons, they show complicated soliton interactions.
In section \ref{sec_rel}, we discuss the relationship between the delay BBS and an already known BBS.
If solitons are normal and a delay parameter is a particular value, the delay BBS is essentially equivalent to the BBS with $K$ kinds of balls (and box-capacity $1$) introduced in~\cite{Takahashi2,Tokihiro1}.
%The delay BBS and the BBS with $K$ kinds of balls generate the same solitons despite the different expressions of their rules.
%However, if solitons are abnormal, the delay BBS generates different solitons from the BBS with $K$ kinds of balls.
Section \ref{sec_con} is devoted to conclusions.

\end{section}

\begin{section}{Construction of a delay \bbs}
\label{sec_cons}

In this section, we derive a delay analogue of the BBS by the \udisn\ of the delay discrete LV equation.
This is inspired by a derivation of the BBS from the \udis\ LV equation~\cite{Tokihiro3}.

We start from the discrete KP equation~\cite{Hirota2,Miwa}:
\begin{eqnarray}
    \label{diskp}
     (b-c)f_{n+1,t,k}f_{n,t+1,k+1}
    +(c-a)f_{n,t+1,k}f_{n+1,t,k+1}\nonumber\\
    \hspace{20mm} +(a-b)f_{n,t,k+1}f_{n+1,t+1,k}
    =0\,.
\end{eqnarray}
We showed in~\cite{Nakata1} that the delay discrete LV equation (\ref{dlydislv_nl}) is obtained by a reduction of the discrete KP equation.
Now, instead of equation (\ref{dlydislv_nl}), we introduce a generalized delay discrete LV equation.
Let us apply the following reduction to equation (\ref{diskp}): 
\begin{equation}
    \label{reduction_lv}
    f_{n,t,k+1}=f_{n+2+\nu,t+\mu,k}\,,
\end{equation}
where the $\mu$ and $\nu$ are fixed integers considered as the delays.
Setting $a=1,\ b=0,\ c=1+\del$ and using (\ref{reduction_lv}), we obtain the discrete bilinear equation:
\begin{equation}
    \label{dlydishlv}
    (1+\del)f_{n+\nu}^{t+1+\mu}f_{n-1}^{t}
    -\del f_{n+1+\nu}^{t+\mu}f_{n-2}^{t+1}
    -f_{n+\nu}^{t+\mu}f_{n-1}^{t+1}
    =0\,,
\end{equation}
where $f_{n}^{t}$ denotes $f_{n,t,0}$.

Next, we apply the dependent variable transformation
\begin{equation*}
    v_{n}^{t}
    =\frac{f_{n+1+\nu}^{t+\mu}f_{n-2}^{t+1}}{f_{n+\nu}^{t+\mu}f_{n-1}^{t+1}}
\end{equation*}
to bilinear equation (\ref{dlydishlv}).
Consequently, we obtain the nonlinear equation
\begin{equation}
    \label{dlydishlv_nl}
    \frac{v_{n+\nu}^{t+1+\mu}v_{n-1}^{t}}{v_{n+\nu}^{t+\mu}v_{n-1}^{t+1}}
    =\frac{(1+\del v_{n+1+\nu}^{t+\mu})(1+\del v_{n-2}^{t+1})}{(1+\del v_{n+\nu}^{t+\mu})(1+\del v_{n-1}^{t+1})}\,.
\end{equation}
%We call (\ref{dlydishlv_nl}) the delay discrete hungry LV equation.
In the case of $\nu=0$, we can check that equation (\ref{dlydishlv_nl}) becomes the delay discrete LV equation (\ref{dlydislv_nl}).
In addition, if $\mu=\nu=0$, equation (\ref{dlydishlv_nl}) becomes a division of the discrete LV equation
\begin{equation*}
    \frac{v_{n}^{t+1}}{v_{n}^{t}}
    =\frac{1+\del v_{n+1}^{t}}{1+\del v_{n-1}^{t+1}}
\end{equation*}
and the space-shifted one
\begin{equation*}
    \frac{v_{n-1}^{t+1}}{v_{n-1}^{t}}
    =\frac{1+\del v_{n}^{t}}{1+\del v_{n-2}^{t+1}}\,.
\end{equation*}

Now, let us consider the \udisn\ of a generalized delay discrete LV equation (\ref{dlydishlv_nl}).
Inspired by the derivation of the \udis\ LV equation in section \ref{sec_intro}, we first apply the following variable transformation to equation (\ref{dlydishlv_nl}):
\begin{equation*}
    \left\{
    \begin{array}{cc}
        n+t&\to t\\
        t&\to n
    \end{array}
    \right.\,,\qquad
    \left\{
    \begin{array}{cc}
        \nu+\mu&= \al\\
        \mu&= \be
    \end{array}
    \right.\,.
\end{equation*}
As a result, equations (\ref{dlydishlv}) and (\ref{dlydishlv_nl}) are transformed into
\begin{equation}
    \label{dlydishlv2}
    (1+\del)f_{n+1+\be}^{t+1+\al}f_{n}^{t-1}
    -\del f_{n+\be}^{t+1+\al}f_{n+1}^{t-1}
    -f_{n+\be}^{t+\al}f_{n+1}^{t}
    =0\,,
\end{equation}
\begin{equation}
    \label{dlydishlv_nl2}
    \frac{v_{n+1+\be}^{t+1+\al}v_{n}^{t-1}}{v_{n+\be}^{t+\al}v_{n+1}^{t}}
    =\frac{(1+\del v_{n+\be}^{t+1+\al})(1+\del v_{n+1}^{t-1})}{(1+\del v_{n+\be}^{t+\al})(1+\del v_{n+1}^{t})}\,.
\end{equation}
Then, we apply the \udisn\ to equation (\ref{dlydishlv_nl2}):
\begin{equation*}
    f_{n}^{t}=\exp\bfrac{F_{n}^{t}}{\eps}\,,\quad
    v_{n}^{t}=\exp\bfrac{V_{n}^{t}}{\eps}\,,\quad
    \del=\exp\bfrac{-1}{\eps}\,,\quad
    \eps\to0\,.
\end{equation*}
Using the formula
\begin{equation*}
    \lim_{\eps\to+0}\eps\log\left(e^{A/\eps}+e^{B/\eps}\right)
    =\max(A,B)
\end{equation*}
for real numbers $A,B$, we obtain the following \udis\ equations
\begin{equation}
    \label{dlyudishlv_bl}
    F_{n+1+\be}^{t+1+\al}+F_{n}^{t-1}
    =\max(F_{n+\be}^{t+1+\al}+F_{n+1}^{t-1}-1,
    F_{n+\be}^{t+\al}+F_{n+1}^{t})\,,
\end{equation}
\begin{eqnarray}
    \label{dlyudishlv}
    &&V_{n+1+\be}^{t+1+\al}
    -V_{n+\be}^{t+\al}
    -V_{n+1}^{t}
    +V_{n}^{t-1}\nonumber\\
    &=&\max\left(V_{n+\be}^{t+1+\al}-1,0\right)
    +\max\left(V_{n+1}^{t-1}-1,0\right)\nonumber\\
    &&-\max\left(V_{n+\be}^{t+\al}-1,0\right)
    -\max\left(V_{n+1}^{t}-1,0\right)\,.
\end{eqnarray}
\begin{equation}
    \label{dlyudishlv_trans}
    V_{n}^{t}
    =F_{n+\be}^{t+1+\al}-F_{n+\be}^{t+\al}-F_{n+1}^{t}+F_{n+1}^{t-1}\,.
\end{equation}
Equation (\ref{dlyudishlv}) is considered as the delay \udis\ LV equation.

We introduce a new dependent variable $U_{n}^{t}$:
% \begin{equation}
%     \label{dlybbs_trans}
%     \color{red}
%     F_{n}^{t}=\sum_{i=t+1}^{\infty}\sum_{j=-\infty}^{n}U_{j}^{i}\,,
% \end{equation}
% \textcolor{red}{which leads to}
% \begin{eqnarray*}
%     \color{red}
%     F_{n}^{t-1}-F_{n}^{t}=\sum_{j=-\infty}^{n}U_{j}^{t}\,,\qquad
%     F_{n}^{t}-F_{n-1}^{t}=\sum_{i=t+1}^{\infty}U_{n}^{i}\,,\\
%     \color{red}
%     U_{n}^{t} = -F_{n}^{t}+F_{n-1}^{t}+F_{n}^{t-1}-F_{n-1}^{t-1}\,.
% \end{eqnarray*}
% \textcolor{red}{Now, equation (\ref{dlyudishlv_bl}) is transformed into}
% \begin{eqnarray*}
%     \color{red}
%     F_{n+\be}^{t+1+\al}-F_{n+1+\be}^{t+1+\al}+F_{n+1}^{t-1}-F_{n}^{t-1}
%     =\min(1,
%     F_{n+\be}^{t+1+\al}-F_{n+\be}^{t+\al}+F_{n+1}^{t-1}-F_{n+1}^{t})\,.
% \end{eqnarray*}
% \textcolor{red}{Using the transformation (\ref{dlybbs_trans}), we obtain}
% \begin{eqnarray*}
%     \color{red}
%     -\sum_{i=t+\al+2}^{\infty}U_{n+\be+1}^{i}
%     +\sum_{i=t}^{\infty}U_{n+1}^{i}
%     =\min\left(1,
%     -\sum_{j=-\infty}^{n+\be}U_{j}^{t+\al+1}
%     +\sum_{j=-\infty}^{n+1}U_{j}^{t}\right)\,,
% \end{eqnarray*}
% \textcolor{red}{which leads to}
\begin{equation}
    \label{trans_lv1}
    V_{n}^{t}
    =\sum_{j=-\infty}^{n+1} U_{j}^{t}
    -\sum_{j=-\infty}^{n+\be} U_{j}^{t+1+\al}
    =(T_1-T_2)\left(\sum_{j=-\infty}^{n} U_{j}^{t}\right)\,,
\end{equation}
where the linear operators $T_1$ and $T_2$ are defined by
\begin{equation*}
    T_1(g_{n}^{t})=g_{n+1}^{t}\,,\qquad
    T_2(g_{n}^{t})=g_{n+\be}^{t+1+\al}
\end{equation*}
for a function $g_{n}^{t}$.
Using (\ref{trans_lv1}), equation (\ref{dlyudishlv}) can be rewritten as
\begin{eqnarray*}
    &&(T_1-T_2)\left(\sum_{j=-\infty}^{n+1+\be} U_{j}^{t+1+\al}
    -\sum_{j=-\infty}^{n+\be} U_{j}^{t+\al}
    -\sum_{j=-\infty}^{n+1} U_{j}^{t}
    +\sum_{j=-\infty}^{n} U_{j}^{t-1}\right)\\
    &=&(T_1-T_2)\left(-\max\left(V_{n}^{t}-1,0\right)
    +\max\left(V_{n}^{t-1}-1,0\right)\right)\,.
\end{eqnarray*}
Thus we can formally obtain the equation
\begin{eqnarray*}
    &&\sum_{j=-\infty}^{n+1+\be} U_{j}^{t+1+\al}
    -\sum_{j=-\infty}^{n+\be} U_{j}^{t+\al}
    -\sum_{j=-\infty}^{n+1} U_{j}^{t}
    +\sum_{j=-\infty}^{n} U_{j}^{t-1}\\
    &=&-\max\left(V_{n}^{t}-1,0\right)
    +\max\left(V_{n}^{t-1}-1,0\right)\,.
\end{eqnarray*}
Using (\ref{trans_lv1}) again, we finally obtain the following equation:
\begin{eqnarray}
    \label{dlybbs_lv1}
    U_{n}^{t+\al+2}&=&U_{n-\be}^{t+1}
    -\min\left(1-U_{n-\be}^{t},\sum_{j=-\infty}^{n-1-\be} U_{j}^{t} - \sum_{j=-\infty}^{n-1} U_{j}^{t+\al+1}\right)\nonumber\\
    &&+\min\left(1-U_{n-\be}^{t+1},\sum_{j=-\infty}^{n-1-\be} U_{j}^{t+1} - \sum_{j=-\infty}^{n-1} U_{j}^{t+\al+2}\right)\,.
\end{eqnarray}
We define equation (\ref{dlybbs_lv1}) as a delay analogue of the BBS.
Equation (\ref{dlyudishlv_bl}) is its bilinear form.
%This equation can be rewritten as
%\begin{eqnarray}
%    \label{dlybbs_lv2}
%    U_{n}^{t+\al+2}&=&U_{n-\be}^{t}
%    -\min\left(1,\sum_{j=-\infty}^{n-\be} U_{j}^{t} - \sum_{j=-\infty}^{n-1} U_{j}^{t+\al+1}\right)\nonumber\\
%    &&+\min\left(1,\sum_{j=-\infty}^{n-\be} U_{j}^{t+1} - \sum_{j=-\infty}^{n-1} U_{j}^{t+\al+2}\right)\,.
%\end{eqnarray}
%We define equation (\ref{dlybbs_lv1}) or (\ref{dlybbs_lv2}) as a delay analogue of the BBS.
%Formulation (\ref{dlybbs_lv1}) is useful for the cases $U_{n}^{t}\in\{0,1\}$, otherwise formulation (\ref{dlybbs_lv2}) is useful.
If $\al=\be=0$, equation (\ref{dlybbs_lv1}) becomes a subtraction of the BBS (\ref{bbs}) and its time-shifted version
\begin{equation*}
    \fl
    U_{n}^{t+1}
    =\min\left(1-U_{n}^{t},\sum_{j=-\infty}^{n-1} U_{j}^{t} - \sum_{j=-\infty}^{n-1} U_{j}^{t+1}\right)\,,\quad
    U_{n}^{t+2}
    =\min\left(1-U_{n}^{t+1},\sum_{j=-\infty}^{n-1} U_{j}^{t+1} - \sum_{j=-\infty}^{n-1} U_{j}^{t+2}\right)\,.
\end{equation*}

\begin{rem}
In this paper, we assume the relation
\begin{equation}
    \label{dlybbs_trans}
    F_{n}^{t}=\sum_{i=t+1}^{\infty}\sum_{j=-\infty}^{n}U_{j}^{i}\,,
\end{equation}
which leads to
\begin{eqnarray*}
    \fl
    F_{n}^{t-1}-F_{n}^{t}=\sum_{j=-\infty}^{n}U_{j}^{t}\,,\qquad
    F_{n}^{t}-F_{n-1}^{t}=\sum_{i=t+1}^{\infty}U_{n}^{i}\,,\qquad
    U_{n}^{t} = -F_{n}^{t}+F_{n-1}^{t}+F_{n}^{t-1}-F_{n-1}^{t-1}\,.
\end{eqnarray*}
This relation (\ref{dlybbs_trans}) is consistent with (\ref{dlyudishlv_trans}) and (\ref{trans_lv1}).
Instituting (\ref{dlybbs_trans}) into the bilinear equation (\ref{dlyudishlv_bl}), we can obtain the delay BBS (\ref{dlybbs_lv1}) without formal operations.
\end{rem}
\begin{rem}
    For the convergence of equation (\ref{dlybbs_lv1}), we impose the boundary condition that $\{n:U_{n}^{t}\neq0\}$ has a lower bound for all $t$.
    When considering the bilinear equation (\ref{dlyudishlv_bl}) and the transformation (\ref{dlybbs_trans}), we add the further condition that  $\{t:U_{n}^{t}\neq0\}$ has an upper bound for all $n$.
\end{rem}

\end{section}

\begin{section}{Observation of soliton patterns of a delay \bbs}
\label{sec_obs}

In this section, we present several soliton interactions of the delay BBS (\ref{dlybbs_lv1}) as numerical experiments and discuss their features.
We note that the delay BBS includes the delay parameter $\al$, and determines the state of time $t+\al+2$ by using the states of times $t,\ t+1,\ t+\al+1$.
Therefore, initial states of times $0$ to $\al+1$ are required for computing the time evolution.

\begin{subsection}{Normal solitons}

Setting up various initial states at multiple times, we can generate various types of solitons.
First, we consider a class of initial states which generates simple soliton patterns.

\begin{Def}
    \label{def_normal}
    We call an initial state \textit{normal} if it satisfies the following conditions.
    \begin{itemize}
        \item [(A)] At each time $0,\ldots,\al+1$, there are $N$ queues that consist of $1$'s:
        More precisely, for all $i=0,\ldots,\al+1$, there exist integers $\tail_i(j),\head_i(j)\ \ (j=1,\ldots,N)$ such that
        \[
        \fl\hspace{20mm}\tail_i(1)\leq\head_i(1)<
        \tail_i(2)\leq\head_i(2)<\ldots<
        \tail_i(N)\leq\head_i(N)\,,
        \]
        and
        \[
        \left\{
        \begin{array}{cl}
            U_{n}^{i}=1&\mbox{if $\ \ \tail_i(j)\leq n\leq\head_i(j)\ \ $ for a $\ \ j=1,\ldots,N$}\\
            U_{n}^{i}=0&\mbox{otherwise}
        \end{array}
        \right.\,.
        \]
        \item [(B)] The tails of the queues are attached to the heads of the previous time queues:
        More precisely, for all $i=1,\ldots,\al+1$ and $j=1,\ldots,N$, it satisfies
        \[\tail_{i}(j)=\head_{i-1}(j)+1\,.\]
        \item [(C)] The lengths of the queues at time $0$ are the same as the lengths of the queues at time $\al+1$:
        More precisely, for all $j=1,\ldots,N$, it satisfies
        \[\head_{0}(j)-\tail_{0}(j)=\head_{\al+1}(j)-\tail_{\al+1}(j)\,.\]
        \item [(D)] The overlaps of the queues at times $0$ and $\al+1$ are no more than $\be$:
        More precisely, for all $j=2,\ldots,N$, it satisfies
        \[\head_{\al+1}(j-1)-\tail_{0}(j)\leq\be-1\,.\]
    \end{itemize}
    Note that if $N=1$, condition (D) is not needed.
\end{Def}

Figure \ref{fig_normal} shows an example of a normal initial state for $\al=2,\ \be=2,\ \ N=2$.
\begin{figure}
    \centering
    \includegraphics%[width=15cm]
    {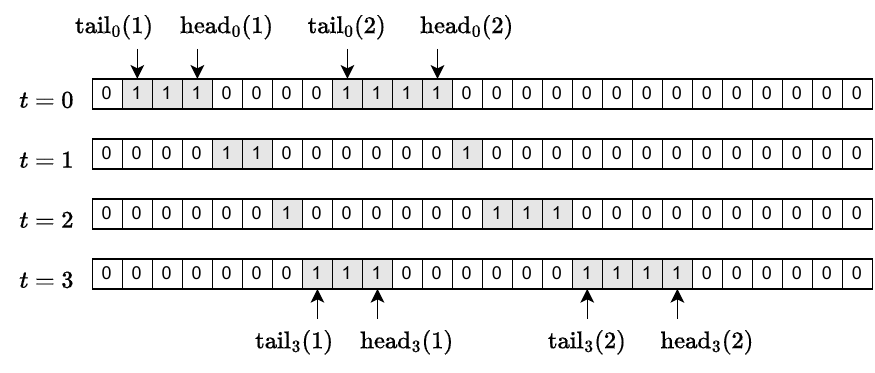}
    \caption{\label{fig_normal}
    An example of a normal initial state for $\al=2,\ \be=2,\ N=2$.}
\end{figure}

Now, let us show several soliton interactions of the delay BBS by using the time evolution equation (\ref{dlybbs_lv1}).
Figures \ref{fig_nml_evo10}, \ref{fig_nml_evo11}, \ref{fig_nml_evo20} and \ref{fig_nml_evo22} show examples of soliton patterns obtained by setting normal initial states.
From these figures, we can observe the speeds of these solitons change with the time period $\al+1$, where speed is defined as the distance travelled in one time.
For example, one of the solitons in figure \ref{fig_nml_evo10} has a two-periodic speed of $5,2,5,2,\cdots$, thus we denote it by $(5,2)$.
Similarly, we can denote the 2-soliton before interaction in figure \ref{fig_nml_evo20} by $(4,3,1)$ and $(1,3,2)$.

According to figures \ref{fig_nml_evo10} and \ref{fig_nml_evo11}, $(5,2)$ and $(1,3)$ deform into $(3,4)$ and $(3,1)$ after interaction.
Similarly, $(4,3,1)$ and $(1,3,2)$ in figures \ref{fig_nml_evo20} and \ref{fig_nml_evo22} deform into $(1,5,2)$ and $(4,1,1)$ after interaction.
The periodicity of speeds and the deforming after interaction are considered to be an effect of the delay.

Another consideration can be derived from these figures.
As we mentioned, $(5,2)$ and $(1,3)$ deform into $(3,4)$ and $(3,1)$ in figure \ref{fig_nml_evo10}, thus we can observe that the average speed of each soliton is conserved ($7/2=(5+2)/2=(3+4)/2$ and $2=(1+3)/2=(3+1)/2$).
The same phenomenon is observed in figures \ref{fig_nml_evo11}, \ref{fig_nml_evo20} and \ref{fig_nml_evo22}.
We discuss this conservation law in section \ref{sec_rule}.

Normal initial states generate soliton patterns commonly observed in known BBSs.
Thus we call solitons obtained from normal initial states \textit{normal solitons}.
A relationship between normal solitons and the BBS with $K$ kinds of balls is discussed in section \ref{sec_rel}.

\begin{figure}
    \centering
    \includegraphics[width=15cm]
    {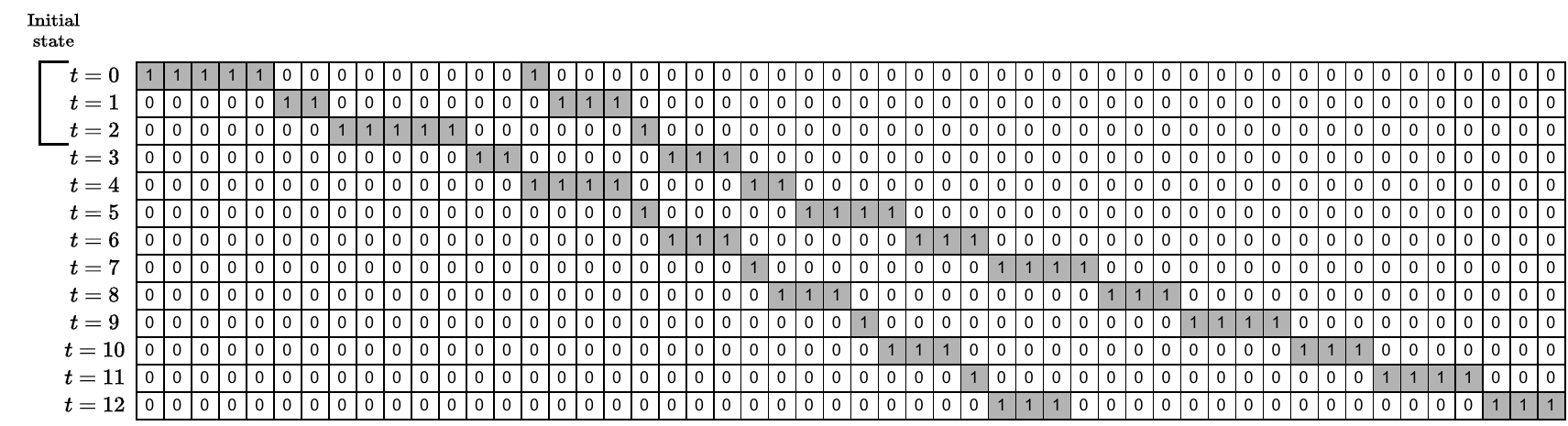}
    \caption{\label{fig_nml_evo10}
    Solitons obtained by normal initial states ($\al=1,\ \be=0,\ N=2$).}
\end{figure}
\begin{figure}
    \centering
    \includegraphics[width=15cm]
    {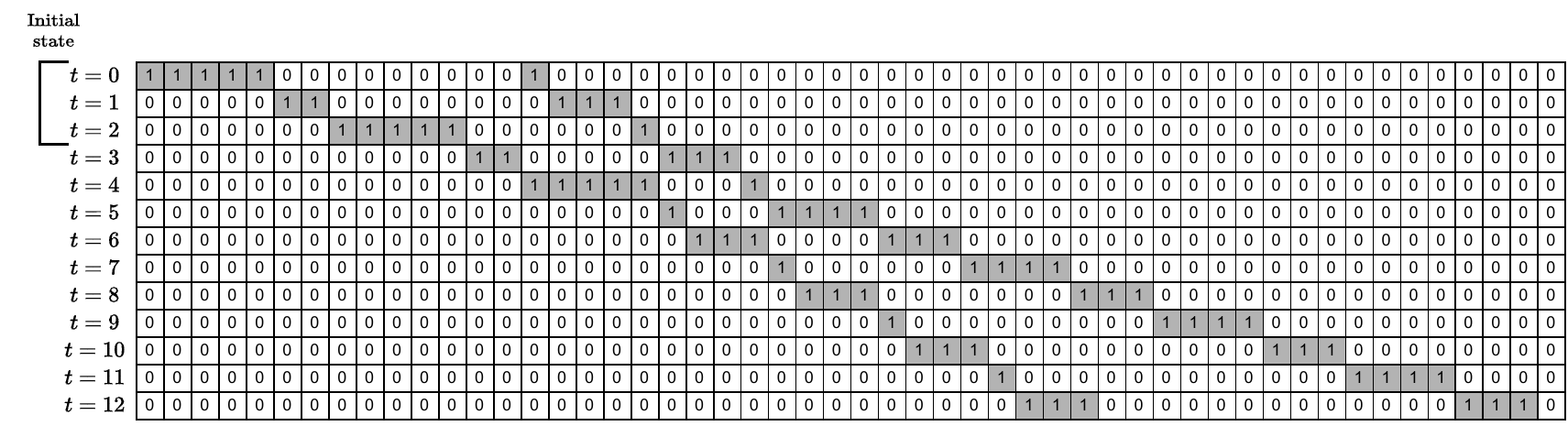}
    \caption{\label{fig_nml_evo11}
    Solitons obtained by normal initial states ($\al=1,\ \be=1,\ N=2$).}
\end{figure}
\begin{figure}
    \centering
    \includegraphics[width=15cm]
    {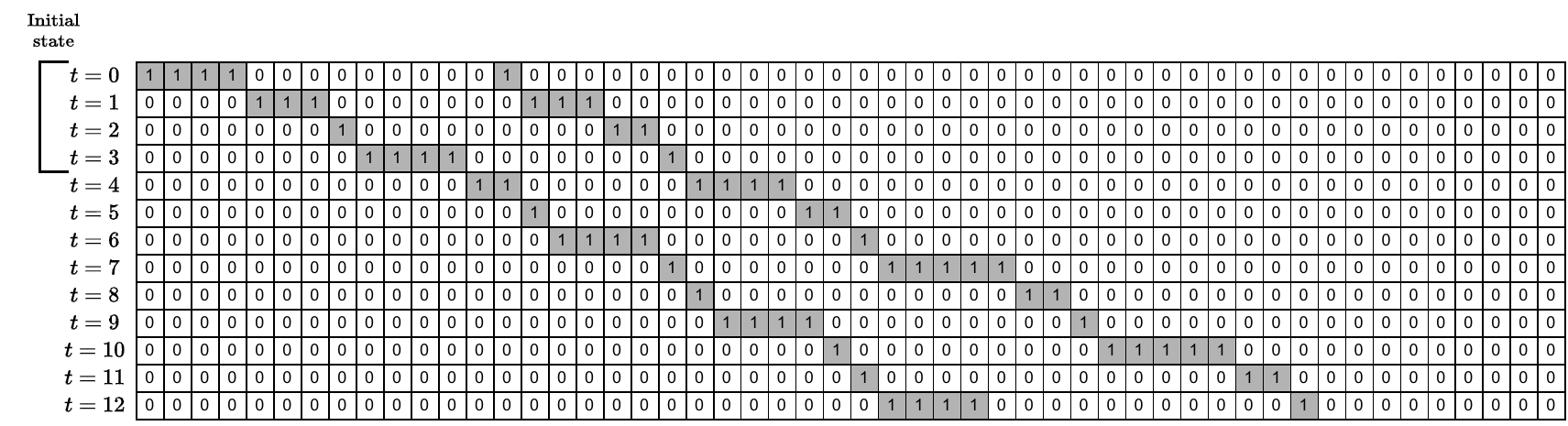}
    \caption{\label{fig_nml_evo20}
    Solitons obtained by normal initial states ($\al=2,\ \be=0,\ N=2$).}
\end{figure}
\begin{figure}
    \centering
    \includegraphics[width=15cm]
    {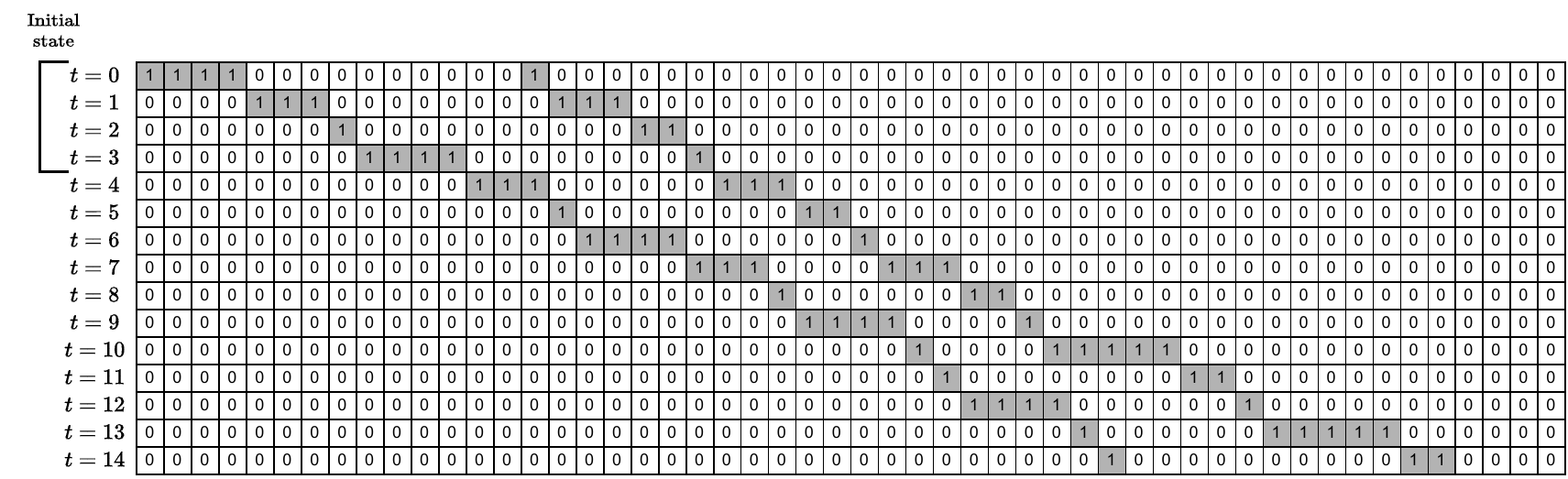}
    \caption{\label{fig_nml_evo22}
    Solitons obtained by normal initial states ($\al=2,\ \be=2,\ N=2$).}
\end{figure}

Finally, note that normal solitons cannot have an empty queue, namely normal solitons $(L,0)$ or $(0,L)$ ($L$ is a natural number) do not exist.
It is because they do not satisfy condition (A) in definition \ref{def_normal}, in addition they cannot keep their shape in the time evolution as we can see from figure \ref{fig_emptyqueue_evo}.
\begin{figure}
    \centering
    \includegraphics[width=12cm]
    {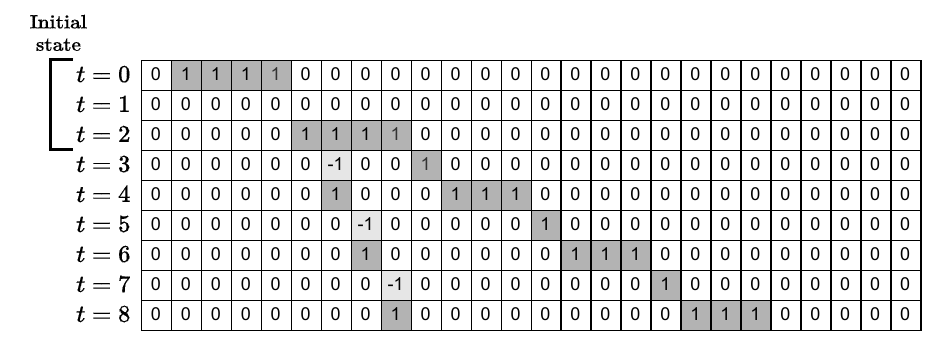}
    \caption{\label{fig_emptyqueue_evo}
    Soliton $(4,0)$, which is not a normal soliton, cannot keep its shape ($\al=1,\ \be=1,\ N=2$).}
\end{figure}

\end{subsection}

\begin{subsection}{Abnormal solitons}

Figure \ref{fig_non_evo} shows soliton interactions of the delay BBS obtained by setting non-normal initial states, which we call \textit{abnormal initial states}.
For example, we can see that the soliton pattern $U_{n}^{t}=U_{n+\be}^{t+\al+1}=1$ appears in figure \ref{fig_non_evo} (a).
%This soliton pattern actually satisfies equation (\ref{dlybbs_lv1}).
We call a soliton with this pattern a \textit{jumping soliton}.
(Because the value of $(t,n)$ jumps to $(t+\al+1,n+\be)$.)
In addition, we can observe other new patterns in figures \ref{fig_non_evo} (b), (c) and (d), which we call a \textit{long-period soliton}, a \textit{negative soliton with speed $1$} and a \textit{triangle pattern} respectively.
We define such solitons which are generated by abnormal initial states as \textit{abnormal solitons}.
%As introduced above, the delay BBS (\ref{dlybbs_lv1}) generates several soliton patterns because of the freedom of multiple time initial states.
\begin{figure}
    \centering
    \includegraphics[width=9cm]
    {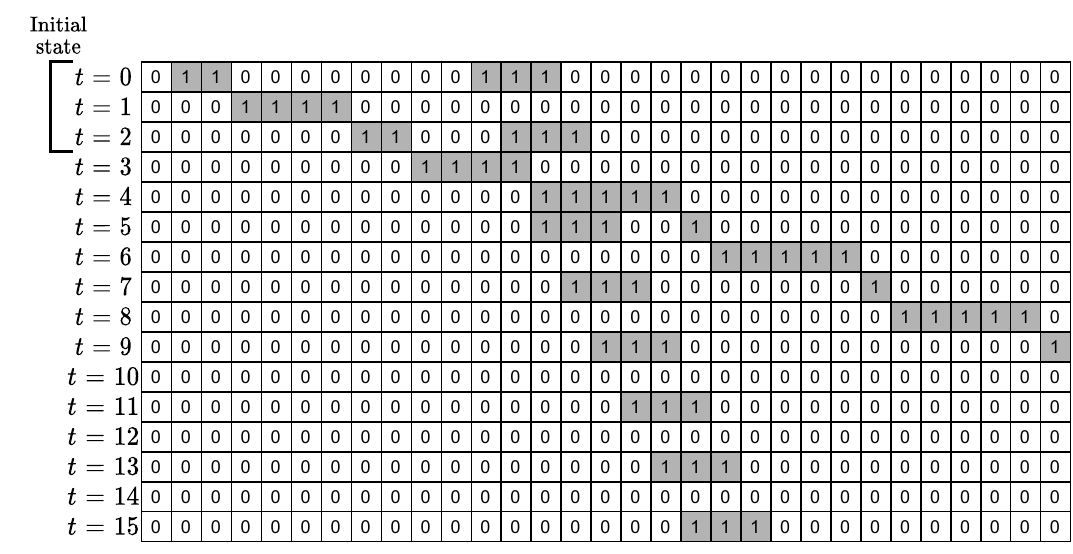}
    \mbox{(a) $\al=1,\ \be=1$.}
    \centering
    \includegraphics[width=9cm]
    {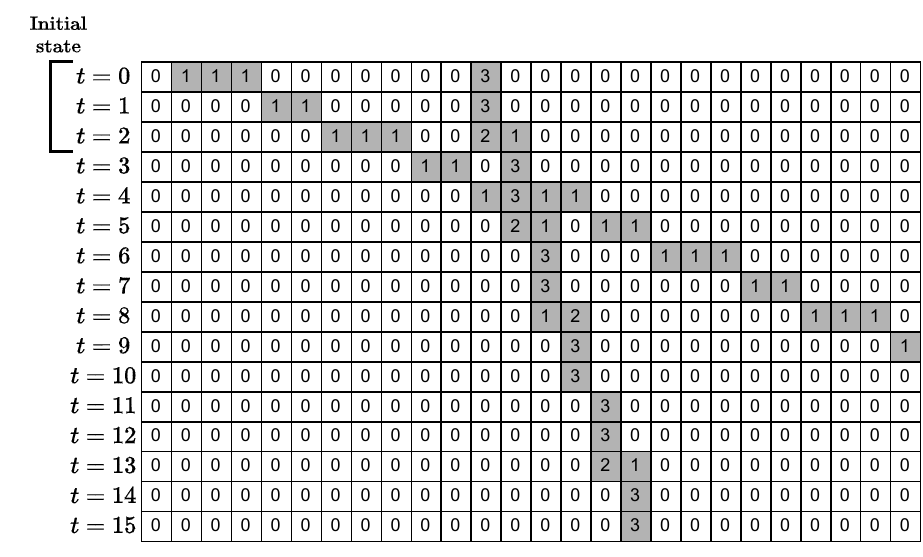}
    \mbox{(b) $\al=1,\ \be=1$.}
    \centering
    \includegraphics[width=9cm]
    {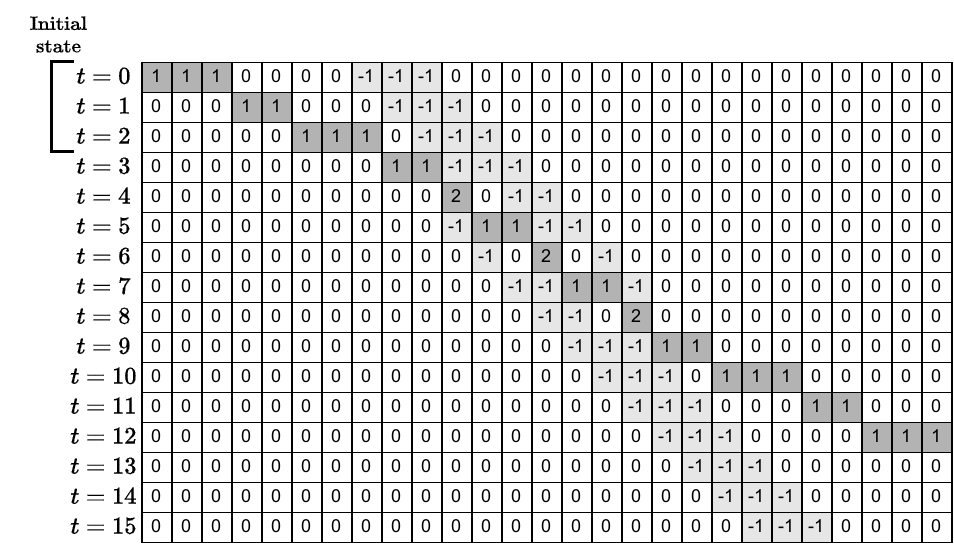}
    \mbox{(c) $\al=1,\ \be=1$.}
    \centering
    \includegraphics[width=9cm]
    {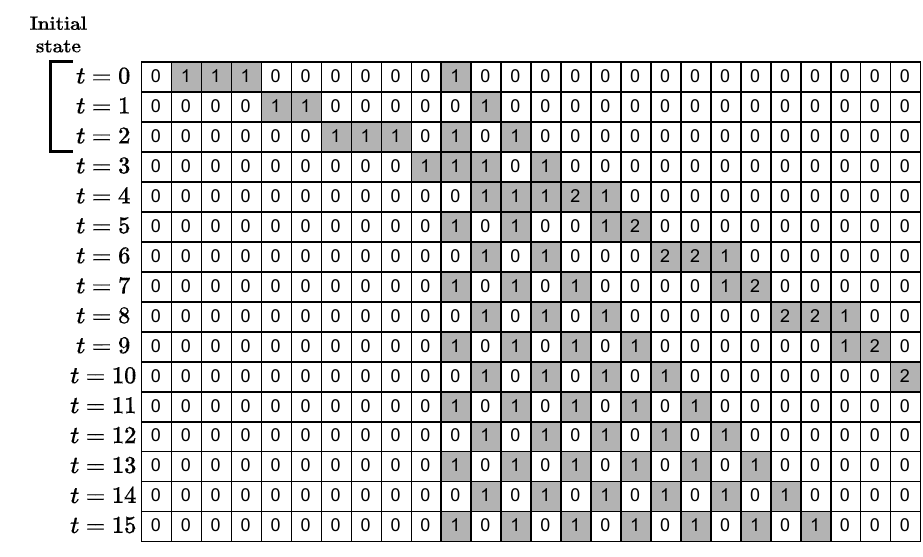}
    \mbox{(d) $\al=1,\ \be=0$.}
\caption{\label{fig_non_evo}
Solitons obtained by abnormal initial states. (The figures show interactions of a normal soliton and an abnormal soliton.)}
\end{figure}

Triangle patterns are observed in the ultra-discrete Toda lattice equation~\cite{Takahashi4}, and a pattern similar to negative solitons with speed $1$ is studied in~\cite{Hirota3}.
However, these abnormal solitons are not observed commonly in known BBSs to the best of our knowledge.

Now, we consider $1$-soliton solutions of them in the bilinear formulation.
From the transformation (\ref{dlybbs_trans}) and an abnormal soliton pattern $U_{n}^{t}$ obtained by a numerical experiment, we can derive a solution $F_{n}^{t}$ for each type of abnormal soliton.
In addition, detailed numerical experiments let us extend the solution to general cases of parameters such as $\al,\be$.
We show $1$-soliton solutions of the jumping and long-period patterns in general forms below.

\subsubsection{The jumping pattern}
\begin{thm}
    \label{thm_jump_element}
    For $\al\geq0$ and $\be\geq1$, the function $J_{n}^{t}$ is defined as follows:
    \begin{eqnarray*}
        J_{n}^{t}=\max\left(0,Cj_{n}^{t}\right)\,,\quad
        j_{n}^{t}=\floor{\frac{a-t}{\al+1}}+\floor{\frac{n-b}{\be}}\,,\quad
        0\leq C\leq1\,,
    \end{eqnarray*}
    where $a,b$ are constant integers, and $C$ is a real constant, and $\floor{x}$ is the floor function, i.e., the maximum integer no more than $x$.
    Then $F_{n}^{t}=J_{n}^{t}$ is a solution of the bilinear equation of the delay BBS (\ref{dlyudishlv_bl}).
\end{thm}
\begin{proof}
    Because the case $C=0$ is obvious, we assume $0<C<1$.
    We first substitute $J_{n}^{t}$ into the bilinear equation (\ref{dlyudishlv_bl}).
    Since $j_{n+\al}^{t+\be}=j_{n}^{t-1}$, the following equation is obtained:
    \begin{equation*}
        J_{n+1}^{t}+J_{n}^{t-1}
        =\max(J_{n}^{t}+J_{n+1}^{t-1}-1,
        J_{n}^{t-1}+J_{n+1}^{t})\,.
    \end{equation*}
    Therefore it is sufficient to prove $J_{n}^{t}+J_{n+1}^{t-1}-1\leq J_{n}^{t-1}+J_{n+1}^{t}$.
    Using the obvious relations
    \begin{equation*}
        j_{n}^{t}\leq j_{n+1}^{t}\leq j_{n+1}^{t-1}\,,\qquad
        j_{n}^{t}\leq j_{n}^{t-1}\leq j_{n+1}^{t-1}\,,
    \end{equation*}
    we need to consider only the following 6 cases:
    \begin{equation*}
    \begin{array}{lllll}
        \dag_1 &J_{n}^{t}=0,\ &J_{n+1}^{t}=0,\ &J_{n}^{t-1}=0,\ &J_{n+1}^{t-1}=0.\\
        \dag_2 &J_{n}^{t}=0,\ &J_{n+1}^{t}=0,\ &J_{n}^{t-1}=0,\ &J_{n+1}^{t-1}=Cj_{n+1}^{t-1}.\\
        \dag_3 &J_{n}^{t}=0,\ &J_{n+1}^{t}=Cj_{n+1}^{t},\ &J_{n}^{t-1}=0,\ &J_{n+1}^{t-1}=Cj_{n+1}^{t-1}.\\
        \dag_4 &J_{n}^{t}=0,\ &J_{n+1}^{t}=0,\ &J_{n}^{t-1}=Cj_{n}^{t-1},\ &J_{n+1}^{t-1}=Cj_{n+1}^{t-1}.\\
        \dag_5 &J_{n}^{t}=0,\ &J_{n+1}^{t}=Cj_{n+1}^{t},\ &J_{n}^{t-1}=Cj_{n}^{t-1},\ &J_{n+1}^{t-1}=Cj_{n+1}^{t-1}.\\
        \dag_6 &J_{n}^{t}=Cj_{n}^{t},\ &J_{n+1}^{t}=Cj_{n+1}^{t},\ &J_{n}^{t-1}=Cj_{n}^{t-1},\ &J_{n+1}^{t-1}=Cj_{n+1}^{t-1}.
    \end{array}
    \end{equation*}
    Cases $\dag_1$ and $\dag_6$ are obvious.
    In case $\dag_2$, it holds that $j_{n+1}^{t}\leq0$ since $J_{n+1}^{t}=0$ and $C>0$.
    Thus we obtain
    \begin{eqnarray*}
        Cj_{n+1}^{t-1}
        &=C\left(\floor{\frac{a-t+1}{\al+1}}+\floor{\frac{n-b+1}{\be}}\right)\\
        &\leq C\left(\floor{\frac{a-t+1}{\al+1}}-\floor{\frac{a-t}{\al+1}}\right)
        \leq C
        \leq1\,.
    \end{eqnarray*}
    Note that we used $\floor{(a-t+1)/(\al+1)}-\floor{(a-t)/(\al+1)}\in\{0,1\}$.
    This leads to
    \begin{equation*}
        J_{n}^{t}+J_{n+1}^{t-1}-1
        =Cj_{n+1}^{t-1}-1
        \leq0
        =J_{n}^{t-1}+J_{n+1}^{t}\,.
    \end{equation*}
    In case $\dag_3$, we obtain
    \begin{eqnarray*}
        Cj_{n+1}^{t-1}-Cj_{n+1}^{t}
        &=C\left(\floor{\frac{a-t+1}{\al+1}}-\floor{\frac{a-t}{\al+1}}\right)
        \leq C
        \leq1\,.
    \end{eqnarray*}
    This leads to
    \begin{equation*}
        J_{n}^{t}+J_{n+1}^{t-1}-1
        =Cj_{n+1}^{t-1}-1
        \leq Cj_{n+1}^{t}
        =J_{n}^{t-1}+J_{n+1}^{t}\,.
    \end{equation*}
    Cases $\dag_4$ and $\dag_5$ can be proved in the same way.
\end{proof}
Figure \ref{fig_jump_evo} shows a jumping soliton described by the above $J_{n}^{t}\ (\al=\be=C=1)$.
\begin{figure}
    \centering
    \includegraphics[width=6cm]
    {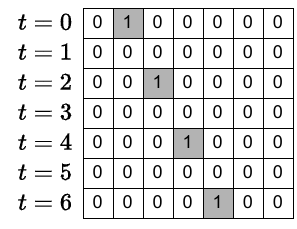}
\caption{\label{fig_jump_evo}
A jumping soliton which corresponds to $J_{n}^{t}$ ($\al=1,\ \be=1,\ C=1$).}
\end{figure}
In this theorem, the constant parameter $C$ corresponds to the amplitude of the jumping soliton.
By numerical experiments, we can realize the amplitude does not need to be a positive number.
Thus we can expect the above $J_{n}^{t}$ to be extended to the case $C<0$.
Actually, the following $J_{n}^{t}$ is a solution of the bilinear equation of the delay BBS (\ref{dlyudishlv_bl}):
\begin{equation*}
    J_{n}^{t}=\min\left(0,Cj_{n}^{t}\right)\,,\quad
    C<0\,,
\end{equation*}
which can be proved in the same way as theorem \ref{thm_jump_element}.
On the other hand, if $C>1$, function $J_{n}^{t}=\max\left(0,Cj_{n}^{t}\right)$ is not a solution of (\ref{dlyudishlv_bl}).
In numerical experiments, no jumping-soliton pattern is observed when the amplitude is greater than 1.

Multiple jumping solitons can be placed as shown in figure \ref{fig_non_evo} (a).
We can extend theorem \ref{thm_jump_element} as follows.
\begin{thm}
    For $\al\geq0$ and $\be\geq1$, the function $\hat{J}_{n}^{t}$ is defined as follows:
    \begin{eqnarray*}
        \fl \hat{J}_{n}^{t}=\sum_{k=1}^{K}E_{n;k}^{t}\,,\quad
        E_{n;k}^{t}=\max\left(0,C_kj_{n;k}^{t}\right)\,,\quad
        j_{n;k}^{t}=\floor{\frac{a_k-t}{\al+1}}+\floor{\frac{n-b_k}{\be}}\,,\\
        0\leq C_k\leq1\,,\quad
        (a_k,b_k)\not\equiv (a_l,b_l)\ (k\neq l)\,,
    \end{eqnarray*}
    where $a_k,b_k,K$ are constant integers, and $C_k$ is a real constant, and $(a,b)\equiv (c,d)$ is defined as follows:
    \begin{quote}
        There exists an integer $m$ such that $a=c+(\al+1)m$ and $b=d+\be m$.
    \end{quote}
    Then $F_{n}^{t}=\hat{J}_{n}^{t}$ is a solution of the bilinear equation of the delay BBS (\ref{dlyudishlv_bl}).
\end{thm}
\begin{proof}
    Similarly to theorem \ref{thm_jump_element}, it is sufficient to prove $\hat{J}_{n}^{t}+\hat{J}_{n+1}^{t-1}-1\leq \hat{J}_{n}^{t-1}+\hat{J}_{n+1}^{t}$.
    According to calculations in theorem \ref{thm_jump_element}, we can obtain
    \begin{equation*}
    \begin{array}{lll}
        \ddag_1
        &E_{n;k}^{t}+E_{n+1;k}^{t-1}-1\leq E_{n;k}^{t-1}+E_{n+1;k}^{t}
        &(j_{n;k}^{t}\leq0\leq j_{n+1;k}^{t-1})\,.\\
        \ddag_2
        &E_{n;k}^{t}+E_{n+1;k}^{t-1}=E_{n;k}^{t-1}+E_{n+1;k}^{t}
        &(0\leq j_{n;k}^{t} \mathrm{\ or\ } j_{n+1;k}^{t-1}\leq0)\,.
    \end{array}
    \end{equation*}
    Case $\ddag_1$ corresponds to $\dag_2$, $\dag_3$, $\dag_4$ and $\dag_5$.
    Case $\ddag_2$ corresponds to $\dag_1$ and $\dag_6$.
    Now, we consider case $\ddag_1$ in more detail.
    Case $\ddag_1$ can be divided into the following 4 cases:
    \begin{eqnarray*}
    \fl
    \begin{array}{lllll}
        \diamondsuit_1 &j_{n;k}^{t}=0,\ &j_{n+1;k}^{t}=0,\ &j_{n;k}^{t-1}=0,\ &j_{n+1;k}^{t-1}=0.\\
        \diamondsuit_2(c_0) &j_{n;k}^{t}=-1+c_0,\ &j_{n+1;k}^{t}=-1+c_0,\ &j_{n;k}^{t-1}=c_0,\ &j_{n+1;k}^{t-1}=c_0.\\
        \diamondsuit_3(c_0) &j_{n;k}^{t}=-1+c_0,\ &j_{n+1;k}^{t}=c_0,\ &j_{n;k}^{t-1}=-1+c_0,\ &j_{n+1;k}^{t-1}=c_0.\\
        \diamondsuit_4(d_0) &j_{n;k}^{t}=-1+d_0,\ &j_{n+1;k}^{t}=d_0,\ &j_{n;k}^{t-1}=d_0,\ &j_{n+1;k}^{t-1}=d_0+1.\\
    \end{array}\\
    (c_0=0,1\,,\quad d_0=-1,0,1)
    \end{eqnarray*}
    In these cases, only $\diamondsuit_4(0)$ leads to
    \begin{equation*}
        E_{n;k}^{t}+E_{n+1;k}^{t-1}-1=E_{n;k}^{t-1}+E_{n+1;k}^{t}\,,
    \end{equation*}
    and the other cases lead to
    \begin{equation*}
        E_{n;k}^{t}+E_{n+1;k}^{t-1}=E_{n;k}^{t-1}+E_{n+1;k}^{t}\,.
    \end{equation*}
    Since case $\diamondsuit_4(0)$ occurs at the points where $(t,n)\equiv(a_k+1,b_k-1)$, we obtain
    \begin{equation*}
    \left\{
    \begin{array}{ll}
        E_{n;k}^{t}+E_{n+1;k}^{t-1}-1=E_{n;k}^{t-1}+E_{n+1;k}^{t}
        &((t,n)\equiv(a_k+1,b_k-1))\\
        E_{n;k}^{t}+E_{n+1;k}^{t-1}=E_{n;k}^{t-1}+E_{n+1;k}^{t}
        &((t,n)\not\equiv(a_k+1,b_k-1))
    \end{array}
    \right.\,.
    \end{equation*}
    Thus we conclude $\hat{J}_{n}^{t}+\hat{J}_{n+1}^{t-1}-1\leq \hat{J}_{n}^{t-1}+\hat{J}_{n+1}^{t}$.
\end{proof}
We note that the jumping soliton in figure \ref{fig_non_evo} (a) is obtained by setting $C_1=C_2=C_3=1,\ K=3,\ \al=1,\ \be=1$.

\begin{rem}
As we mentioned in section \ref{sec_cons}, in the case of $\al=\be=0$, the delay BBS (\ref{dlybbs_lv1}) is just a subtraction of the two BBSs.
However, this subtracted BBS allows a soliton pattern which is not observed in the well-known BBS, which is the jumping solitons $U_{n}^{t}=U_{n}^{t+1}=C$ ($C\leq 1$) shown in figure \ref{fig_nml_jump_evo1}.
\begin{figure}
    \centering
    \includegraphics[width=12cm]
    {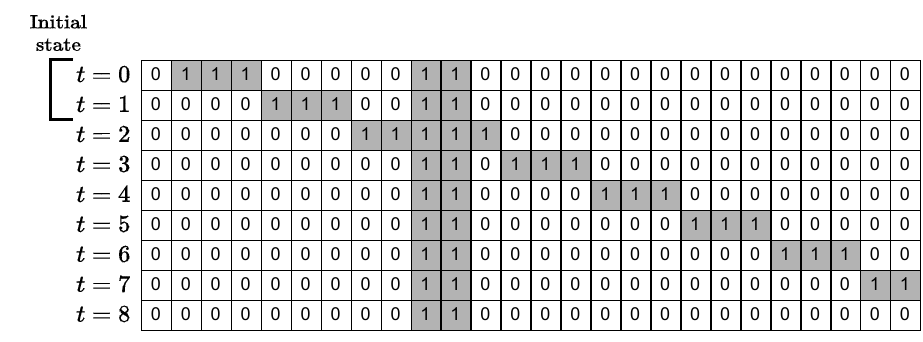}
    \caption{\label{fig_nml_jump_evo1}
    A normal soliton and a jumping soliton ($\al=0,\ \be=0$).}
\end{figure}
\end{rem}

\subsubsection{The long-period pattern}
\begin{thm}
    \label{thm_longperiod}
    For $\al\geq0$ and $1\leq\be\leq\al+1$, the function $L_{n}^{t}$ is defined as follows:
    \begin{eqnarray*}
        L_{n}^{t}=\max\left(0,l_{n}^{t}\right)\,,\quad
        l_{n}^{t}=A(a-t)+B(n-b)\,,\quad
        B=\frac{(\al+2)A-1}{\be}\,,\quad
        A\geq1\,,
    \end{eqnarray*}
    where $a,b,A$ are constant integers.
    Then $F_{n}^{t}=L_{n}^{t}$ is a solution of the bilinear equation of the delay BBS (\ref{dlyudishlv_bl}).
\end{thm}
\begin{proof}
    Substituting $L_{n}^{t}$ into the bilinear equation (\ref{dlyudishlv_bl}) and using the relations
    \begin{eqnarray*}
        \fl
        L_{n}^{t-1}=\max\left(0,l_{n}^{t}+A\right)\,,\quad
        L_{n+1}^{t}=\max\left(0,l_{n}^{t}+B\right)\,,\quad
        L_{n+\be}^{t+\al}=\max\left(0,l_{n}^{t}+2A-1\right)\,,
    \end{eqnarray*}
    equation (\ref{dlyudishlv_bl}) is transformed into
    \begin{eqnarray}
        \eqalign{
        \label{longperiod_bl}
        &\max\left(0,l_{n}^{t}+A+B-1\right)
        +\max\left(0,l_{n}^{t}+A\right)\\
        =\max(&\max\left(0,l_{n}^{t}+2A-1\right)
        +\max\left(0,l_{n}^{t}+B\right),\\
        &\max\left(0,l_{n}^{t}+A-1\right)
        +\max\left(0,l_{n}^{t}+A+B\right)-1)\,.
        }
    \end{eqnarray}
    Since $1\leq\be\leq\al+1$,
    \begin{equation*}
        B=\frac{(\al+2)A-1}{\be}\geq\frac{(\al+2)A-1}{\al+1}=A+\frac{A-1}{\al+1}\geq A\,.
    \end{equation*}
    Thus we need to consider only the following 8 cases:
    \begin{equation*}
    \begin{array}{lrl}
        \clubsuit_1 & l_{n}^{t}+A+B\leq0.\\
        \clubsuit_2 & l_{n}^{t}+A+B-1\leq0, & 0\leq l_{n}^{t}+A+B.\\
        \clubsuit_3 & l_{n}^{t}+B,\ l_{n}^{t}+2A-1\leq0, & 0\leq l_{n}^{t}+A+B-1.\\
        \clubsuit_4 & l_{n}^{t}+B\leq0, & 0\leq l_{n}^{t}+2A-1.\\
        \clubsuit_5 & l_{n}^{t}+2A-1\leq0, & 0\leq l_{n}^{t}+B.\\
        \clubsuit_6 & l_{n}^{t}+A\leq0, & 0\leq l_{n}^{t}+B,\ l_{n}^{t}+2A-1.\\
        \clubsuit_7 & l_{n}^{t}+A-1\leq0, & 0\leq l_{n}^{t}+A.\\
        \clubsuit_8 & & 0\leq l_{n}^{t}+A-1.
    \end{array}
    \end{equation*}
    For example, in case $\clubsuit_2$, equation (\ref{longperiod_bl}) is $0=\max\left(0, l_{n}^{t}+A+B-1\right)$.
    Since the condition $l_{n}^{t}+A+B-1\leq0$, this equation holds true.
    Other cases can be proved similarly.
\end{proof}
We note that the long-period soliton in figure \ref{fig_non_evo} (b) is obtained by setting $A=3,\ \al=1,\ \be=1$.

The above theorem can be extended as follows:
\begin{thm}
\label{thm_longperiod_extend}
    For $\al\geq0$ and $\be\geq1$, the integers $T,N,I,J$ are uniquely defined as
    \begin{eqnarray}
        \label{TNIJ}
        \eqalign{
        a-t=(\al+1)T+I\quad &(0\leq I\leq\al)\,,\\
        n-b=\be N+J\quad &(0\leq J\leq\be-1)\,,
        }
    \end{eqnarray}
    where $a,b$ is a constant integer.
    In addition the function $L_{n}^{t}$ is defined as follows:
    \begin{eqnarray*}
        L_{n}^{t}=\max\left(0,l_{n}^{t}\right)\,,\quad
        l_{n}^{t}=\tilde{A}T+\sum_{i=1}^{I}A_i
        +\tilde{B}N+\sum_{j=1}^{J}B_j\,,\\
        A_i=A\geq1\quad (1\leq i\leq\al+1)\,,\\
        \tilde{A}=\sum_{i=1}^{\al+1}A_i=(\al+1)A\,,\\
        B_j\geq A\quad (1\leq j\leq\be)\,,\\
        \tilde{B}=\sum_{j=1}^{\be}B_j=(\al+2)A-1\,,
    \end{eqnarray*}
    where $A,B_j$ are constant integers.
    Then $F_{n}^{t}=L_{n}^{t}$ is a solution of the bilinear equation of the delay BBS (\ref{dlyudishlv_bl}).
\end{thm}
This theorem can be proved in the same way as theorem \ref{thm_longperiod} by using $L_{n+1}^{t}=\max\left(0,l_{n}^{t}+B_{J+1}\right)$ or $L_{n+1}^{t}=\max\left(0,l_{n}^{t}+B_{0}\right)$.
We remark that theorem \ref{thm_longperiod} is the case
\begin{eqnarray*}
    B_j=\frac{(\al+2)A-1}{\be}\quad (1\leq j\leq\be)
\end{eqnarray*}
of theorem \ref{thm_longperiod_extend}.

\begin{rem}
\label{rem_form}
The jumping 1-soliton solution $J_{n}^{t}$ and the long-period 1-soliton solution $L_{n}^{t}$ can be represented in the same form.
The jumping 1-soliton solution $J_{n}^{t}$ in theorem \ref{thm_jump_element} is rewritten as follows:
\begin{eqnarray*}
    J_{n}^{t}=\max\left(0,Cj_{n}^{t}\right)\,,\quad
    j_{n}^{t}=\tilde{A}T+\sum_{i=1}^{I}A_i
    +\tilde{B}N+\sum_{j=1}^{J}B_j\,,\\
    A_i=
    \left\{
    \begin{array}{ll}
        0 & (1\leq i\leq\al)\\
        1 & (i=\al+1)
    \end{array}
    \right.\,,\\
    \tilde{A}=\sum_{i=1}^{\al+1}A_i=1\,,\\
    B_j=
    \left\{
    \begin{array}{ll}
        0 & (1\leq j\leq\be-1)\\
        1 & (j=\be)
    \end{array}
    \right.\,,\\
    \tilde{B}=\sum_{j=1}^{\be}B_j=1\,,
\end{eqnarray*}
where $T,N,I,J$ are defined as (\ref{TNIJ}).

This representation is a natural extension of the representation used in~\cite{Tokihiro2}.
Future studies should clarify the most general forms of the exact solutions including the multi-soliton solutions in order to fully discuss abnormal solitons.
\end{rem}
\color{black}

\end{subsection}

\end{section}

\begin{section}{The rule for the time evolution}
\label{sec_rule}

\begin{subsection}{Normal solitons for $\be=0$}
\label{subsec_rule_nml_beta0}

In this section, we consider the rule for the time evolution of the delay BBS (\ref{dlybbs_lv1}).
First, we assume that initial states are normal and $\be=0$.
%$U_{n}^{t}\in\{0,1\}$ for all $n$ and $t$
Considering equation (\ref{dlybbs_lv1}) on the normal solitons, we can compute the state of time $t+\al+2$ from those of times only $t+1,\ t+\al+1$ by the following rule:
\begin{itemize}
    \item [(i)] We define the function $P_{n}^{t+\al+2}$ as follows:
    \begin{eqnarray*}
        \fl\hspace{20mm}\left\{
        \begin{array}{ll}
            P_{n}^{t+\al+2}=-1&\mbox{$\head_{t+1}(j)+1\leq n\leq\head_{t+\al+1}(j),\ \ j=1,\ldots,N$}\\
            P_{n}^{t+\al+2}=0&\mbox{otherwise}
        \end{array}
        \right.\,.
    \end{eqnarray*}
    \item [(ii)] The carrier moves from the $-\infty$ site to the $\infty$ site on the time $t+1$.
    When the carrier passes the $n$th box, we apply the following process:
    \begin{itemize}
        \item If there is a ball ($U_{n}^{t+1}=1$), the carrier gets this ball from the box.
        \item If there is no ball ($U_{n}^{t+1}=0$) and the carrier carries at least $1$ ball, the carrier puts a ball into the box.
    \end{itemize}
    However, if $P_{n}^{t+\al+2}=-1$, the above process is not applied and the carrier passes the box without doing anything (the function $P$ means ``prohibited'').
    \item [(iii)] The resulting state is the state of time $t+\al+2$.
\end{itemize}

Now, we can reconstruct the normal solitons shown in figures \ref{fig_nml_evo10} and \ref{fig_nml_evo11}.
Using the above rule, we obtain them again as shown in figures \ref{fig_nml_evo_rule1} and \ref{fig_nml_evo_rule2}.
\begin{figure}
    \centering
    \includegraphics
    {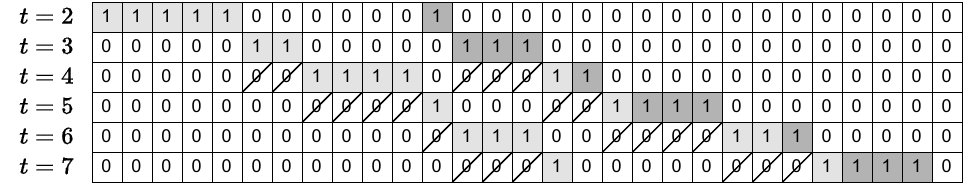}
    \caption{\label{fig_nml_evo_rule1}
    Normal solitons obtained by the rule ($\al=1,\ \be=0,\ N=2$).
    Diagonal lines on the boxes mean the values of the function $P$ are $-1$.} 
\end{figure}
\begin{figure}
    \centering
    \includegraphics
    {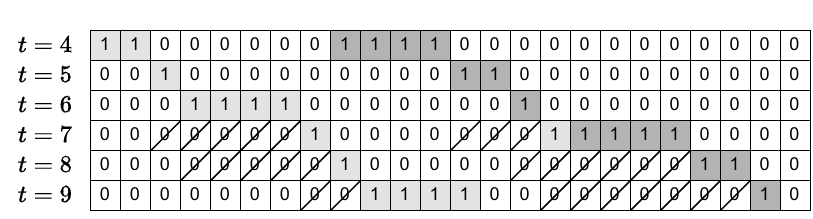}
    \caption{\label{fig_nml_evo_rule2}
    Normal solitons obtained by the rule ($\al=2,\ \be=0,\ N=2$).
    Diagonal lines on the boxes mean the values of the function $P$ are $-1$.}
\end{figure}

\begin{rem}
Equation (\ref{dlybbs_lv1}) for the case $\be=0$ can be rewritten as
\begin{eqnarray*}
    U_{n}^{t+\al+2}=P_{n}^{t+\al+2}+\min\left(1-U_{n}^{t+1},\sum_{j=-\infty}^{n-1} U_{j}^{t+1} - \sum_{j=-\infty}^{n-1} U_{j}^{t+\al+2}\right)\,.
\end{eqnarray*}
Thus, if $\min\left(1-U_{n}^{t+1},\sum_{j=-\infty}^{n-1} U_{j}^{t+1} - \sum_{j=-\infty}^{n-1} U_{j}^{t+\al+2}\right)=1$, but if $P_{n}^{t+\al+2}=-1$, then $U_{n}^{t+\al+2}=0$ (namely a ball is not put into the $n$th box).
This can be interpreted as that if $P_{n}^{t+\al+2}=-1$, putting a ball into the $n$th box is prohibited.
\end{rem}

As we can see from figures \ref{fig_nml_evo10} and \ref{fig_nml_evo11} (or \ref{fig_nml_evo_rule1} and \ref{fig_nml_evo_rule2}), normal solitons for $\be=0$ travel forward satisfying conditions (A), (B), and (D) of definition \ref{def_normal} for all the times.
More precisely, the following properties hold true for normal solitons.
    \begin{itemize}
        \item [(A')] For all $t=0,1,2,\ldots$, there exist the locations $\tail_t(j),\head_t(j)\ \ (j=1,\ldots,N)$ such that
        \[
        \fl\hspace{20mm}\tail_t(1)\leq\head_t(1)<
        \tail_t(2)\leq\head_t(2)<\ldots<
        \tail_t(N)\leq\head_t(N)\,,
        \]
        and
        \[
        \left\{
        \begin{array}{cl}
            U_{n}^{t}=1&\mbox{if $\ \ \tail_t(j)\leq n\leq\head_t(j)\ \ $ for a $\ \ j=1,\ldots,N$}\\
            U_{n}^{t}=0&\mbox{otherwise}
        \end{array}
        \right.\,.
        \]
        \item [(B')] For all $t=1,2,\ldots$ and $j=1,\ldots,N$, it satisfies
        \[\tail_{t}(j)=\head_{t-1}(j)+1\,.\]
        \item [(D')$_{0}$] For all $t=0,1,2,\ldots$ and $j=2,\ldots,N$, it satisfies
        \[\head_{t+\al+1}(j-1)+1\leq\tail_{t}(j)\,.\]
        In other words, each queue at time $t+\al+1$ does not overlap with the front soliton's queue at time $t$.
        Note that this property holds only when $\be=0$.
    \end{itemize}
(A'), (B'), (D')$_{0}$ are important properties of the delay BBS for $\be=0$.
This is because the above rule for the time evolution can be described more easily by using these properties.
For example, in figure \ref{fig_nml_evo_rule1}, we can consider that the back soliton's queue at time 2 (length 5) is divided into the lengths 4 and 1 at time 4 in order that property (D')$_{0}$ is satisfied.

In general, the rule for the time evolution (i), (ii) and (iii) can be described easily as follows.
\begin{itemize}
    \item [(I)] The queues at time $t$ are copied at time $t+\al+1$ so that (B') is satisfied.
    \item [(II)$_{0}$] The generated queues at time $t+\al+1$ are split and moved to the front soliton's queue in order that (D')$_{0}$ is satisfied (namely, in order that the generated queues do not overlap with the front soliton's queue at time $t$).
    Note that this rule holds only when $\be=0$.
    %More precisely, when (I) is applied to $j-1$th queue at time $t+\al+1$ and if $\head_{t+\al+1}(j-1)\geq\tail_{t}(j)$, then we split the queue at $\head_{t+\al+1}(j-1)$
\end{itemize}

\begin{eg}
    The time evolution shown in figure \ref{fig_nml_evo_rule1} is explained as follows.
    \begin{itemize}
        \item [$t=4$] The back soliton's queue at time 2 (length 5) is copied and divided into the lengths 4 and 1 by rule (II)$_{0}$, to avoid the overlap with the front soliton's queue at time 2.
        \item [$t=5$] The back soliton's queue at time 3 (length 2) is copied and divided into the lengths 1 and 1 by rule (II)$_{0}$, to avoid the overlap with the front soliton's queue at time 3.
        \item [$t=6$] The back soliton's queue at time 4 (length 4) is copied and divided into the lengths 3 and 1 by rule (II)$_{0}$, to avoid the overlap with the front soliton's queue at time 4.
        \item [$t=7$] The back soliton's queue at time 5 (length 1) is copied, and (II)$_{0}$ is not applied since the generated queue does not overlap with the front soliton's queue at time 5.
    \end{itemize}
    The time evolution shown in figure \ref{fig_nml_evo_rule2} is explained as follows.
    \begin{itemize}
        \item [$t=7$] The back soliton's queue at time 4 (length 2) is copied and divided into the lengths 1 and 1 by rule (II)$_{0}$, to avoid the overlap with the front soliton's queue at time 4.
        \item [$t=8$] The back soliton's queue at time 5 (length 1) is copied, and (II)$_{0}$ is not applied since the generated queue does not overlap with the front soliton's queue at time 5.
        \item [$t=9$] The back soliton's queue at time 6 (length 4) is copied, and (II)$_{0}$ is not applied since the generated queue does not overlap with the front soliton's queue at time 6.
    \end{itemize}
\end{eg}

Now, we show an interesting property about the normal solitons.
According to rule (II)$_{0}$, solitons $(5,2)$ and $(1,3)$ interact differently depending on the distance of the solitons set by the initial state.
However, they always deform into $(3,4)$ and $(3,1)$ after any interactions of $(5,2)$ and $(1,3)$.
In general, we can prove the following theorem.
\begin{thm}
    \label{thm_interaction}
    We assume $\al=1$, $\be=0$, and $x_1,x_2,y_1,y_2$ are integers such that
    \[x_1,x_2,y_1,y_2\geq1\,,\qquad x_1+x_2\geq y_1+y_2\,.\]
    The normal solitons $(x_1,x_2)$ and $(y_1,y_2)$ deform as follows after interaction.
    \begin{itemize}
        \item [$\spadesuit_1$] If $x_1\geq y_2$ and $x_2\geq y_1$, the solitons deform into $(x_1+y_1-y_2, x_2-y_1+y_2)$ and $(y_2,y_1)$.
        \item [$\spadesuit_2$] If $x_1\geq y_2$ and $x_2\leq y_1$, the solitons deform into $(x_1+x_2-y_2, y_2)$ and $(y_1-x_2+y_2,x_2)$.
        \item [$\spadesuit_3$] If $x_1\leq y_2$ and $x_2\geq y_1$, the solitons deform into $(y_1, x_2+x_1-y_1)$ and $(x_1,y_2-x_1+y_1)$.
    \end{itemize}
    This interaction law does not depend on the distance of solitons set by the initial state.
\end{thm}
\begin{proof}
    We first consider case $\spadesuit_1$.
    We assume that $T$ is the time when the interaction starts, namely rule (II)$_{0}$ is first applied.
    Without loss of generality, we can assume the queue with length $x_1$ is divided into the lengths $l_1$ and $x_1-l_1$ at time $T$ by rule (II)$_{0}$.
    %Then we can obtain figure \ref{fig_interaction}.
    Here note that $l_1\geq y_2$ because if $l_1<y_2$ then the state of time $T-1$ is against rule (II)$_{0}$, i.e.,
    $\tail_{T-3}(2)\leq\head_{T-1}(1)$.
    %the back soliton's queue at time $T-1$ overlaps the front soliton's queue at time $T-3$.
    Now let us denote $\head_{t}(k)-\tail_{t}(k)$ by $\ell_{t}(k)$.
    Since $x_2\geq y_1$, the back soliton's queue with length $x_2$ at time $T-1$ is divided into the lengths $y_1$ and $x_2-y_1$ at time $T+1$ by rule (II)$_{0}$.
    This leads to $\ell_{T+1}(1)=y_1$ and $\ell_{T+1}(2)=y_2+(x_2-y_1)$ (This is illustrated in figure \ref{fig_interaction}).
    Similarly, since $l_1\geq y_2$, the back soliton's queue with length $l_1$ at time $T$ is divided into the lengths $y_2$ and $l_1-y_2$ at time $T+2$ by rule (II)$_{0}$.
    This leads to $\ell_{T+2}(1)=y_2$ and $\ell_{T+2}(2)=\ell_{T}(2)+(l_1-y_2)=y_1+x_1-y_2$.
    As shown in figure \ref{fig_interaction}, rule (II)$_{0}$ is never used after time $T+3$.
    Thus case $\spadesuit_1$ has been proved.
    The same method can be used to prove other cases.
\end{proof}
\begin{figure}
    \centering
    \includegraphics[width = 18cm]
    {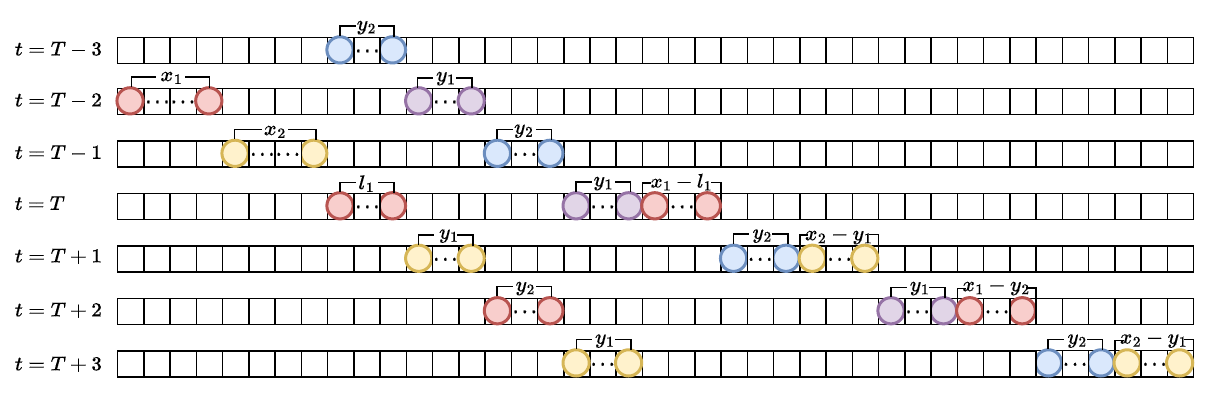}
    \caption{\label{fig_interaction}
    Interaction of $(x_1,x_2)$ and $(y_1,y_2)$ (the case of $x_1\geq y_2$ and $x_2\geq y_1$).}
\end{figure}

This theorem immediately leads to several conservation laws.
\begin{cor}
    \label{cor_conservation}
    We consider the situation of theorem \ref{thm_interaction}.
    Let $(u_1,u_2)$ and $(v_1,v_2)$ be the solitons which $(x_1,x_2)$ and $(y_1,y_2)$ deform after interaction.
    Then it holds that
    \begin{eqnarray}
        \label{conservation1}
        u_1+u_2&=&x_1+x_2\,,\\
        \label{conservation2}
        v_1+v_2&=&y_1+y_2\,,\\
        \label{conservation3}
        u_1+v_1&=&x_1+y_1\,,\\
        \label{conservation4}
        u_2+v_2&=&x_2+y_2\,.
    \end{eqnarray}
\end{cor}
Here, (\ref{conservation1}) and (\ref{conservation2}) are interpreted as the average speed of each soliton being conserved, which is mentioned in section \ref{sec_obs}.
(\ref{conservation3}) and (\ref{conservation4}) can be interpreted as mass preservation.

\end{subsection}

\begin{subsection}{Normal solitons for general $\be$}

Let us consider arbitrary $\be$ such that $\be\leq\al+1$, and assume the initial state is normal.
In this case, we can prove properties (A'), (B'), and (D')$_{\be}$ by the same discussion as $\be=0$.
Here property (D')$_{\be}$ is as follows.
    \begin{itemize}
        \item [(D')$_{\be}$] For all $t=0,1,2,\ldots$ and $j=2,\ldots,N$, it satisfies
        \[\head_{t+\al+1}(j-1)-\tail_{t}(j)\leq\be-1\,.\]
        That is, the overlap between behind the time $t+\al+1$ of each soliton and ahead from the time $t$ of the front soliton is no more than $\be$.
    \end{itemize}
Now, the rule for the time evolution can be extended as follows.
\begin{itemize}
    \item [(I)] The queues at time $t$ are copied at time $t+\al+1$ so that (B') is satisfied.
    \item [(II)$_{\be}$] The generated queues at time $t+\al+1$ are split and moved to the front soliton's queue in order that (D')$_{\be}$ is satisfied (namely, in order that the overlap between behind the time $t+\al+1$ of each soliton and ahead from the time $t$ of the front soliton is no more than $\be$).
\end{itemize}

Figures \ref{fig_nml_evo_rule3} and \ref{fig_nml_evo_rule4} show normal soliton interactions for $\be\neq0$.
In figure \ref{fig_nml_evo_rule3}, we can observe the back soliton's queue at time 3 (length 2) is divided into the lengths 1 and 1 at time 5 according to rule (II)$_{1}$.
Similarly in figure \ref{fig_nml_evo_rule4}, we can observe the back soliton's queue at time 10 is divided into the lengths 1 and 2 according to rule (II)$_{2}$.
\begin{figure}
    \centering
    \includegraphics[width=12cm]
    {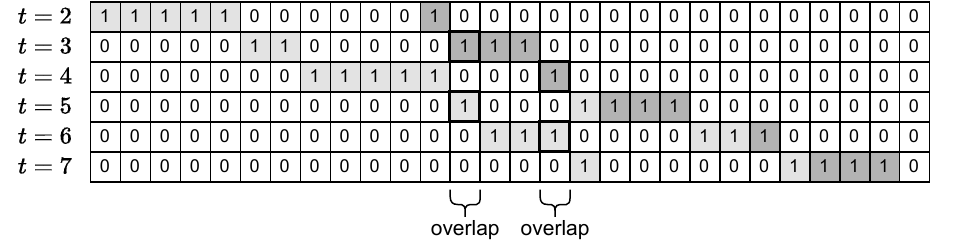}
    \caption{\label{fig_nml_evo_rule3}
    A normal soliton obtained by the rule ($\al=1,\ \be=1,\ N=2$).}
\end{figure}
\begin{figure}
    \centering
    \includegraphics[width=10cm]
    {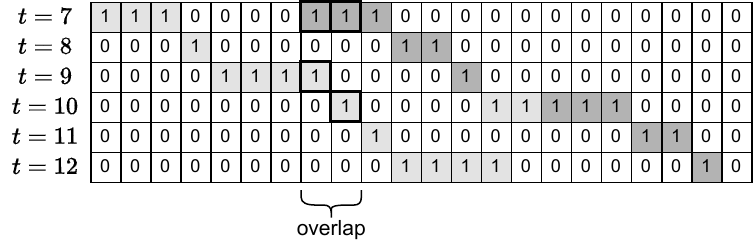}
    \caption{\label{fig_nml_evo_rule4}
    A normal soliton obtained by the rule ($\al=2,\ \be=2,\ N=2$).}
\end{figure}

By the same discussion as $\be=0$, we can prove theorem \ref{thm_interaction} (therefore, corollary \ref{cor_conservation}) for arbitrary $\be\leq\al+1$.
%It means that the same interaction law and several conservation laws hold for arbitrary $\be\leq\al+1$.

Interestingly, if $\al+1<\be$, the above discussion does not necessarily hold true.
For example, when $\al=1$ and $\be=8$, the delay BBS sometimes generates periodic patterns as shown in figure \ref{fig_nml_evo_periodicPattern}.
These periodic patterns consist of some normal solitons, jumping solitons and negative solitons with speed $1$.
Thus in this case, abnormal solitons appear from a normal initial state.
\begin{figure}
    \centering
    \includegraphics[width=12cm]
    {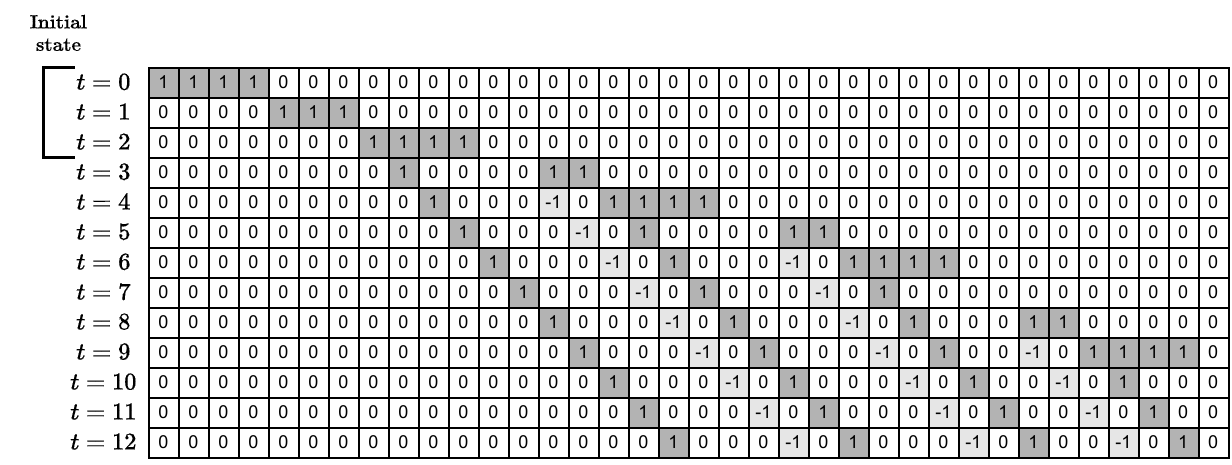}
\caption{\label{fig_nml_evo_periodicPattern}
A soliton pattern in the case of $\al<\be$ ($\al=1,\ \be=8$).
Periodic patterns are generated.}
\end{figure}

\end{subsection}

\begin{subsection}{Interaction of normal and jumping solitons}

Let us extend our discussion to abnormal solitons.
In contrast to normal solitons, abnormal ones sometimes generate complicated interactions which are difficult to analyze.

\subsubsection{Easy cases}

Let us first consider easy cases.
Figure \ref{fig_nml_jump_evo2-1-1} shows interactions of a normal soliton and the most simple jumping solitons in the case of $\al=1,\ \be=0$.
We can observe these solitons deform after interactions.
%From this figure, when a normal soliton and a jumping soliton interact, we can observe that the state of the jumping soliton is copied to the next time.
%For example, the state at time 5 of the left jumping soliton (i.e. value 1) is copied to time 6.
%Similarly, the state at time 9 of the right jumping soliton (i.e. value 0) is copied to time 10.
%and the phase is shifted thereafter.

Now, let us denote by $[1,0]$ the left jumping soliton in figure \ref{fig_nml_jump_evo2-1-1}, and by $[0,1]$ the right jumping soliton.
The difference between $[1,0]$ and $[0.1]$ is the time the balls exist.
%In addition, if we interpret the most simple jumping solitons as $(1,0)$ or $(0,1)$ properly, the interaction law in theorem \ref{thm_interaction} holds in this case.
%(Note that $(1,0)$ or $(0,1)$ does not exist as a normal soliton as we mentioned in section \ref{sec_obs}.)
Considering the time evolution equation (\ref{dlybbs_lv1}), we can show the following interaction law:

\begin{thm}
    \label{thm_interaction_ex}
    We assume $\al=1$, $\be\leq\al+1=2$, and $x_1,x_2$ are integers.
    \begin{itemize}
        \item If $x_1\geq1,\ x_2\geq2$, the normal solitons $(x_1,x_2)$ and jumping soliton $[1,0]$ deform into normal soliton $(x_1+1, x_2-1)$ and jumping soliton $[0,1]$ after interaction.
        \item If $x_1\geq2,\ x_2\geq1$, the normal solitons $(x_1,x_2)$ and jumping soliton $[0,1]$ deform into normal soliton $(x_1-1, x_2+1)$ and jumping soliton $[1,0]$ after interaction.
    \end{itemize}
    This interaction law does not depend on the distance of solitons set by the initial state.
\end{thm}

According to this theorem, we can understand the interaction in figure \ref{fig_nml_jump_evo2-1-1}, namely solitons $(1,3)$ and $[1,0]$ deform into $(2,2)$ and $[0,1]$, and then solitons $(2,2)$ and $[0,1]$ deform into $(1,3)$ and $[1,0]$.

This discussion can be extended to the case of a jumping soliton which consists of only value $1$.
%Such a jumping soliton can be considered as a combination of $[1,0]$ or $[0,1]$.
For example, the jumping soliton in figure \ref{fig_nml_jump_evo2-1-2} consists of three $[1,0]$ and two $[0,1]$, which we call jumping solitons $A[1,0]$, $B[0,1]$, $B[1,0]$, $C[0,1]$ and $C[1,0]$, respectively.
According to theorem \ref{thm_interaction_ex}, we can divide the interaction in figure \ref{fig_nml_jump_evo2-1-2} into the following five interactions of normal and jumping solitons:
\begin{eqnarray*}
               &(1,3),\ A[1,0],\ B[0,1],\ B[1,0],\ C[0,1],\ C[1,0]\\
    \rightarrow&A[0,1],\ (2,2),\ B[0,1],\ B[1,0],\ C[0,1],\ C[1,0]\\
    \rightarrow&A[0,1],\ B[1,0],\ (1,3),\ B[1,0],\ C[0,1],\ C[1,0]\\
    \rightarrow&A[0,1],\ B[1,0],\ B[0,1],\ (2,2),\ C[0,1],\ C[1,0]\\
    \rightarrow&A[0,1],\ B[1,0],\ B[0,1],\ C[1,0],\ (1,3),\ C[1,0]\\
    \rightarrow&A[0,1],\ B[1,0],\ B[0,1],\ C[1,0],\ C[0,1],\ (2,2)
\end{eqnarray*}
\begin{figure}
    \centering%(1,3)が(1,0)の後(0,1)と衝突
    \includegraphics[width=10cm]
    {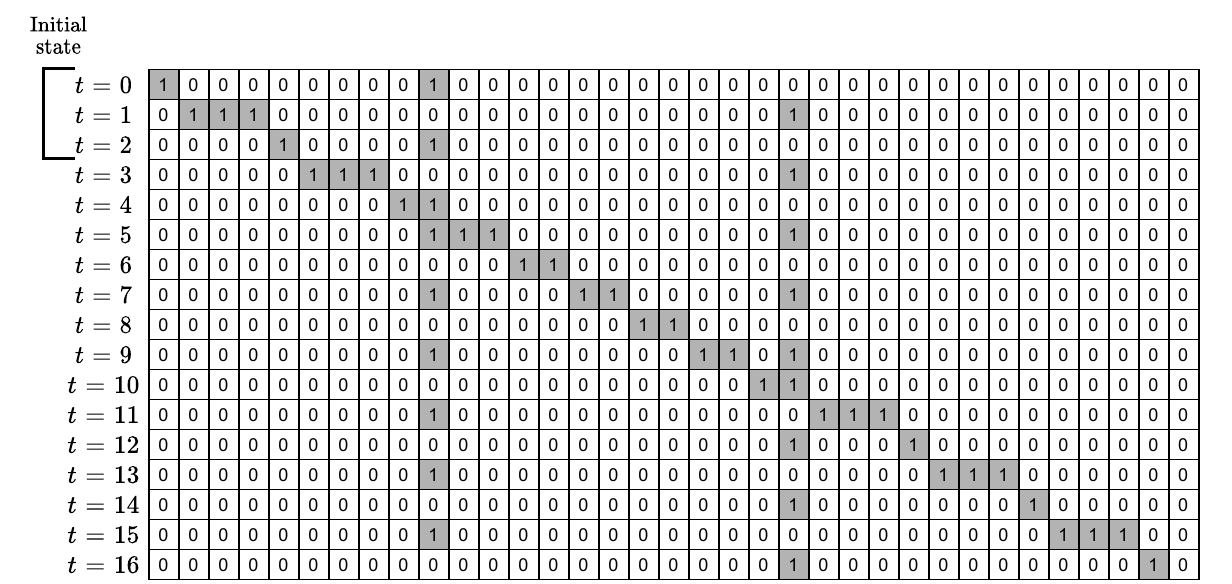}
\caption{\label{fig_nml_jump_evo2-1-1}
An interaction of normal and jumping solitons ($\al=1,\ \be=0$).}
\end{figure}
\begin{figure}
    \centering%値が1のみのjumping soliton
    \includegraphics[width=10cm]
    {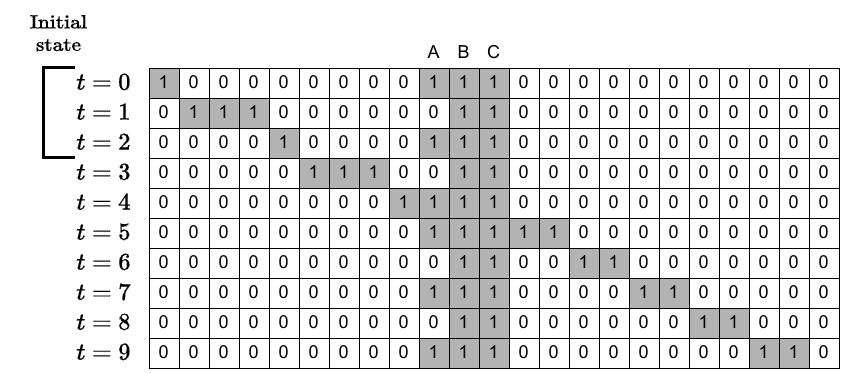}
\caption{\label{fig_nml_jump_evo2-1-2}
An interaction of normal and jumping solitons ($\al=1,\ \be=0$).}
\end{figure}

\subsubsection{Complicated cases}

We move the discussion to several complicated cases.
First, we consider the interaction of a normal soliton $(1,3)$ and a jumping soliton $[0,1]$ shown in figure \ref{fig_nml_jump_evo2-2-1}.
The interaction law in theorem \ref{thm_interaction_ex} does not hold in this case.
If it holds in this case, a normal soliton $(0,4)$ is produced after the interaction, however $(0,4)$ does not exist as a normal soliton as mentioned in section \ref{sec_obs}.
Note that this strange phenomenon is caused by abnormal solitons.
%does not occur in the interactions of normal solitons.
%An interaction of normal solitons produces only normal solitons.

Figure \ref{fig_nml_jump_evo2-2-2} and \ref{fig_nml_jump_evo2-2-3} show interactions of three or more normal and jumping normal solitons.
We can observe that values composing the jumping solitons change after these interactions.
These phenomena should be explained by developing the discussion in this subsection.

\begin{figure}
    \centering
    \includegraphics[width=10cm]
    {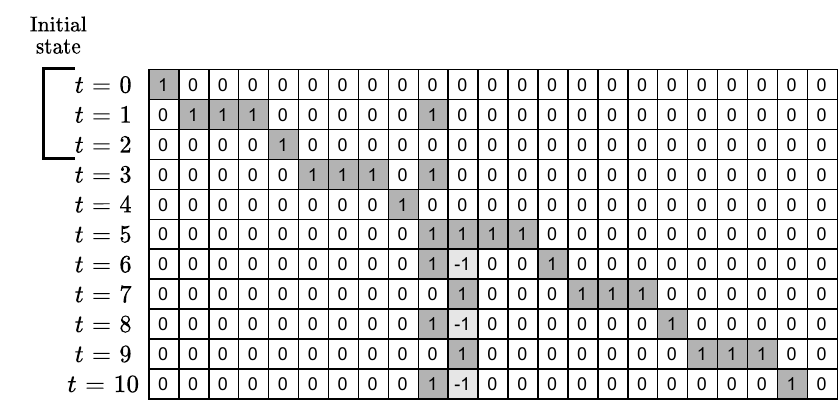}
\caption{\label{fig_nml_jump_evo2-2-1}
An interaction of normal and jumping solitons ($\al=1,\ \be=0$).}
\end{figure}
\begin{figure}
    \centering
    \includegraphics[width=10cm]
    {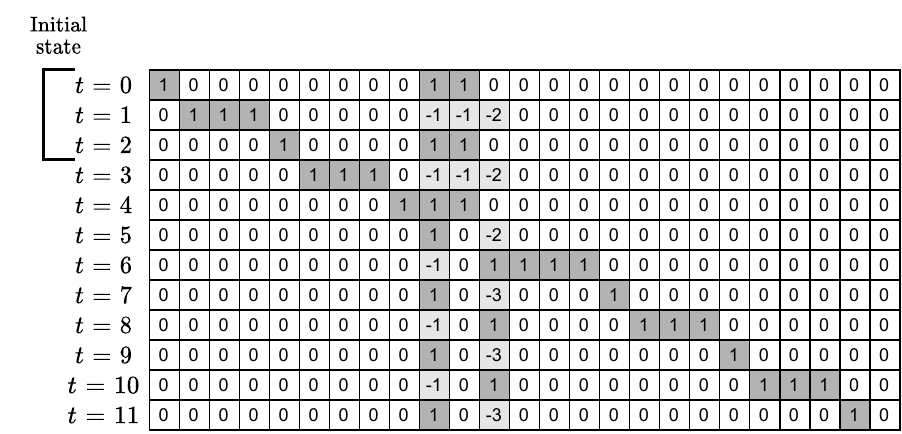}
\caption{\label{fig_nml_jump_evo2-2-2}
An interaction of normal and jumping solitons ($\al=1,\ \be=0$).}
\end{figure}
\begin{figure}
    \centering
    \includegraphics[width=10cm]
    {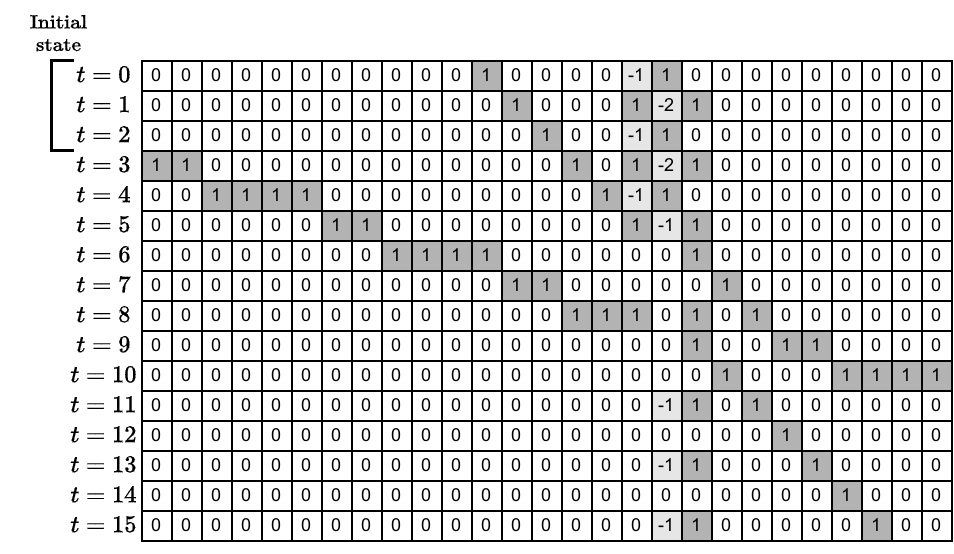}
\caption{\label{fig_nml_jump_evo2-2-3}
An interaction of normal and jumping solitons ($\al=1,\ \be=0$).}
\end{figure}

\end{subsection}

\end{section}

\begin{section}{The relationship to the BBS with $K$ kinds of balls}
\label{sec_rel}

In this section, we reveal the relationship between the delay BBS and an already known BBS.
As we discussed in previous sections, if initial states are normal, the delay BBS generates solitons which have $(\al+1)$-periodic speeds.
We can show that interactions of such normal solitons generated by the delay BBS for $\be=0$ are essentially equivalent to soliton interactions generated by the BBS with $\al+1$ kinds of balls (and box-capacity $1$) introduced in~\cite{Takahashi2}.

Figure \ref{fig_nml_evo4} (which is same as figure \ref{fig_nml_evo20}) shows normal solitons $(4,3,1)$ and $(1,3,2)$ generated by the delay BBS for $\al=2,\ \be=0$.
The time evolution of these solitons is described by rule (I) and (II)$_{0}$ introduced in subsection \ref{subsec_rule_nml_beta0}.
On the other hand, figure \ref{fig_kind3} shows solitons generated by the BBS with $3$ kinds of balls.
The time evolution rule for this system is to first move the numbered 1 ball according to the BBS rule, then the numbered 2 ball, and finally the numbered 3 ball from time $t$ to $t+1$.
%and repeat this process.
\begin{figure}
    \centering
    \includegraphics[width=17cm]
    {fig4.pdf}
    \caption{\label{fig_nml_evo4}
    Solitons obtained by normal initial states ($\al=2,\ \be=0$).}
    \centering
    \includegraphics[width=17cm]
    {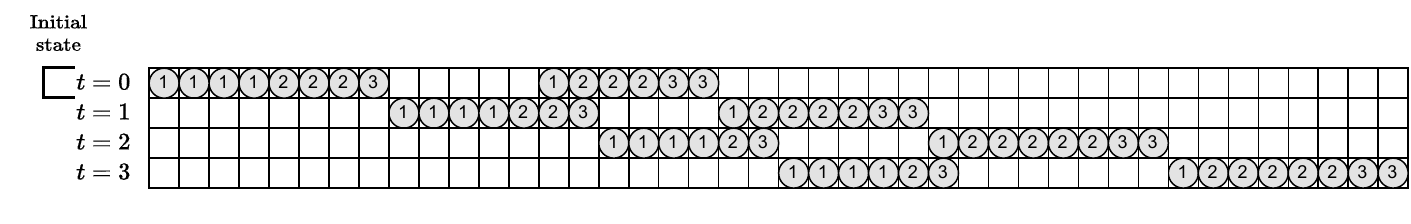}
    \caption{\label{fig_kind3}
    Solitons generated by the BBS with $3$ kinds of balls.}
\end{figure}

The above soliton obtained by these two types of BBSs are equivalent, despite the different representations of the rules.
The states of time 0, 1 and 2 in figure \ref{fig_nml_evo4} correspond to the numbered 1, 2 and 3 balls of time 0 in figure \ref{fig_kind3} respectively.
Similarly, the states of time 3, 4 and 5 in figure \ref{fig_nml_evo4} correspond to the numbered 1, 2 and 3 balls of time 1 in figure \ref{fig_kind3} respectively.

This correspondence can be easily extended to the general $\al$ case.
Namely, one cycle ($\al+1$ time periods) of the delay BBS for $\be=0$ corresponds to one time period of the BBS with $\al+1$ kinds of balls.

Now, we prove this correspondence by the bilinear formulation.
As shown in~\cite{Tokihiro2}, the bilinear equation of the BBS with $\al+1$ kinds of balls is
\begin{equation}
    F_{n+1}^{t+1+\al}+F_{n}^{t-1}
    =\max(F_{n}^{t+1+\al}+F_{n+1}^{t-1}-1,
    F_{n+1}^{t+\al}+F_{n}^{t})\,.
\end{equation}
On the other hand, according to (\ref{dlyudishlv_bl}), the bilinear equation of the delay BBS when $\be=0$ is
\begin{equation*}
    F_{n+1}^{t+1+\al}+F_{n}^{t-1}
    =\max(F_{n}^{t+1+\al}+F_{n+1}^{t-1}-1,
    F_{n}^{t+\al}+F_{n+1}^{t})\,.
\end{equation*}
Although these two bilinear equations differ at the last term in the right-hand sides, the above correspondence exists between them.
Now, we show these two bilinear equations are essentially equivalent in the case of the normal solitons.
In other words, the following theorem is proved.
\begin{thm}
    Let $\al\geq0$ and $a_{n}^{t},b_{n}^{t},c_{n}^{t}$ be as follows:
    \begin{eqnarray*}
        a_{n}^{t}=F_{n+1}^{t+\al}+F_{n}^{t}\,,\quad
        b_{n}^{t}=F_{n}^{t+\al}+F_{n+1}^{t}\,,\quad
        c_{n}^{t}=F_{n}^{t+1+\al}+F_{n+1}^{t-1}\,.
    \end{eqnarray*}
    If $U_{n}^{t} = -F_{n}^{t}+F_{n-1}^{t}+F_{n}^{t-1}-F_{n-1}^{t-1}$ is a normal soliton, then the relation
    \begin{equation*}
        a_{n}^{t}=b_{n}^{t}
    \end{equation*}
    or inequalities
    \begin{equation*}
        a_{n}^{t}<c_{n}^{t}\,,\quad
        b_{n}^{t}<c_{n}^{t}
    \end{equation*}
    hold true.
\end{thm}
\begin{proof}
    Since the case of $\al=0$ is obvious, we assume $\al\geq1$.
    It is sufficient to prove the following statement:
    \begin{quote}
        If $a_{n}^{t}\neq b_{n}^{t}$, then $a_{n}^{t}<c_{n}^{t}$ and $b_{n}^{t}<c_{n}^{t}$.
    \end{quote}
    By using $U_{n}^{t} = -F_{n}^{t}+F_{n-1}^{t}+F_{n}^{t-1}-F_{n-1}^{t-1}$, the following relation is obtained:
    \begin{equation*}
        b_{n}^{t}-a_{n}^{t}
        =-F_{n+1}^{t+\al}+F_{n}^{t+\al}+F_{n+1}^{t}-F_{n}^{t}
        =U_{n+1}^{t+\al}+U_{n+1}^{t+\al-1}+\ldots+U_{n+1}^{t+1}\,.
    \end{equation*}
    Since $U_{n}^{t}$ is a normal soliton, $U_{n}^{t}$ is in $\{0,1\}$ for all $n,t$.
    Thus we obtain
    \begin{equation*}
        a_{n}^{t}\leq b_{n}^{t}
    \end{equation*}
    and
    \begin{equation*}
        a_{n}^{t}\neq b_{n}^{t}
        \Leftrightarrow (\exists p\in\{1,2,\ldots,\al\}
        \mathrm{\ such\ that\ }
        U_{n+1}^{t+p}=1)\,.
    \end{equation*}
    Therefore we need to prove the following statement:
    \begin{quote}
        If $(\exists p\in\{1,2,\ldots,\al\}
        \mathrm{\ such\ that\ }
        U_{n+1}^{t+p}=1)$, then $b_{n}^{t}<c_{n}^{t}$.
    \end{quote}
    Now, using the relation
    \begin{equation*}
        F_{n}^{t-1}-F_{n}^{t}=\sum_{j=-\infty}^{n}U_{j}^{t}\,,
    \end{equation*}
    we obtain
    \begin{equation*}
        b_{n}^{t}<c_{n}^{t}\Leftrightarrow
        \sum_{j=-\infty}^{n}U_{j}^{t+\al+1}<\sum_{j=-\infty}^{n+1}U_{j}^{t}\,.
    \end{equation*}
    Since $U_{n}^{t}$ is a normal soliton, we can easily show
    \begin{equation*}
        \sum_{j=-\infty}^{n}U_{j}^{t+\al+1}
        \leq\sum_{j=-\infty}^{n}U_{j}^{t}
        \leq\sum_{j=-\infty}^{n+1}U_{j}^{t}\,.
    \end{equation*}
    Thus
    \begin{equation*}
        b_{n}^{t}<c_{n}^{t}\Leftrightarrow
        \sum_{j=-\infty}^{n}U_{j}^{t+\al+1}\neq\sum_{j=-\infty}^{n+1}U_{j}^{t}\,.
    \end{equation*}
    In order to prove the statement we need, let us assume
    \begin{equation*}
        U_{n+1}^{t+p_1}=1\,,\quad p_1\in\{1,2,\ldots,\al\}
    \end{equation*}
    and
    \begin{equation}
        \label{assume1}
        \sum_{j=-\infty}^{n}U_{j}^{t+\al+1}=\sum_{j=-\infty}^{n+1}U_{j}^{t}\,,
    \end{equation}
    and show the contradiction.
    By (\ref{assume1}), we can show that
    \begin{equation}
        \label{property1}
        \fl\head_{t}(k)\leq n+1\Rightarrow
        \head_{t+p}(k)\leq n\quad \forall p\in\{1,2,\ldots,\al+1\}\ \forall k\in\{1,2,\ldots\}\,,
    \end{equation}
    as shown in figure \ref{fig_nml_form}.
    On the other hand, since $U_{n+1}^{t+p_1}=1$, there exists $k_1$ such that
    \begin{equation*}
        \head_{t}(k_1)<\tail_{t+p_1}(k_1)\leq n+1\leq\head_{t+p_1}(k_1) 
    \end{equation*}
    This is inconsistent because it contradicts (\ref{property1}).
    \begin{figure}
        \centering
        \includegraphics[width=15cm]
        {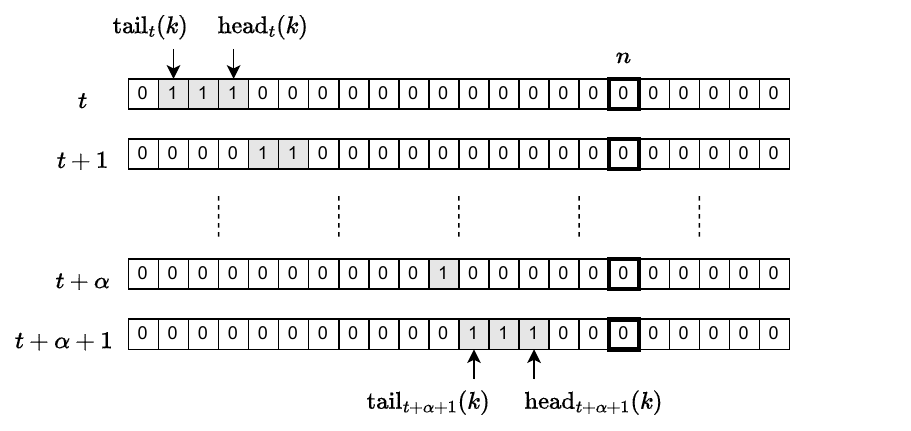}
        \caption{\label{fig_nml_form}
        A normal soliton satisfying (\ref{assume1}).}
    \end{figure}
\end{proof}

Note that several differences exist between the two systems.
When $\be\neq0$ or solitons are abnormal, soliton solutions of the delay BBS do not correspond to those of the BBS with $\al+1$ kinds of balls, such as figure \ref{fig_nml_evo22} and \ref{fig_non_evo}.
%In contrast, the soliton solution of the BBS with $\al+1$ kinds of balls shown in figure \ref{fig_kind3} does not correspond to that of the delay BBS.

\end{section}

\begin{section}{Conclusions}
\label{sec_con}

We have constructed a delay analogue of the BBS by the \udisn\ of the delay discrete LV equation.
This delay BBS includes the delay parameter $\al$, and determines the state of time $t+\al+2$ by using the states of times $t,\ t+1,\ t+\al+1$.
As we mentioned in section \ref{sec_obs}, we need initial states of times $0$ to $\al+1$ for the time evolution.
This feature produces a variety of soliton patterns such as normal and abnormal solitons.
Normal solitons have $(\al+1)$-periodic speeds, which are considered to be effects of the delay.
Abnormal solitons show soliton patterns not found in known BBSs.
We constructed their $1$-soliton solutions in the bilinear formulation.

In section \ref{sec_rule}, the time evolution equation of the delay BBS was interpreted into an elementary rule described by boxes and balls in the case of the normal solitons.
In addition, interactions of normal and jumping solitons are considered.

As we discussed in section \ref{sec_rel}, solitons generated by the delay BBS for $\be=0$ and the BBS with $\al+1$ kinds of balls become exactly the same when initial states are normal.
We showed that the bilinear equations of both BBSs are essentially equivalent when solitons are normal.
%On the other hand, solitons of the delay BBS (the case $\be\neq0$ or abnormal solitons) are not expressed by the BBS with $K$ kinds of balls.

%In section \ref{sec_cons}, the dependent variable $U_{n}^{t}$ was introduced by the transformation (\ref{trans_lv1}).
%Actually, using another transformation, we can construct another type of delay BBS.
%Analysis of this delay BBS will be presented in our forthcoming paper.

In addition to the above, we consider that some important problems with the delay BBS remain.
First, multi-solitons obtained numerically in this paper have not yet been described by explicit soliton solutions.
%Although there are difficulties in finding multi-soliton solutions, we need to construct the exact solutions of them, especially $2$-solitons of normal and abnormal solitons.
We need to construct multi-soliton solutions, especially multi-abnormal soliton solutions.
In addition, we need to construct the solution of the triangle pattern as shown in figure \ref{fig_non_evo} (d).
The peculiarity of this pattern is that the sum $\sum_{n=-\infty}^\infty U_n^t$ (the total of the dependent variables across the spatial dimension) increases with the time evolution.
Second, it may be interesting to consider the delay BBS with a carrier whose capacity is finite, as a delay extension of the BBS with a finite-capacity carrier introduced in~\cite{Takahashi2}.
Third, in this paper, we have not given the mathematical definition to each type of abnormal soliton such as the jumping soliton and long-period soliton.
%We should address this problem, and fully understand the abnormal solitons.
This problem should be investigated, for example, by advancing the considerations discussed in Remark \ref{rem_form}.

Finally, relationships between the delay BBS and the delay discrete Toda or delay discrete KdV equations~\cite{Tsunematsu,Nakata1,Nakata2} should be clarified.
Because some studies revealed that the BBS is related to the discrete Toda molecule and discrete KdV equations~\cite{Nagai,Mada}.
We address these problems in future studies.

\end{section}

\ack
The authors would like to thank Prof.~Daisuke Takahashi for his helpful comments on this work.
This work was partially supported by JSPS KAKENHI Grant Numbers 22K03441, 22H01136 and Waseda University Grants for Special Research Projects.

\end{document}